\documentclass[12pt]{article}
\usepackage[T1]{fontenc}
\usepackage{amsfonts,amsmath,amsthm,amssymb,verbatim,bm}
\usepackage[bottom]{footmisc}
\usepackage{chngcntr}
\usepackage[pdftex,letterpaper,left=1.2in,top=1in,bottom=1in,right=1.2in]{geometry}
\usepackage[round,longnamesfirst]{natbib}
\usepackage{graphicx, color}
\graphicspath{{Figures/}} 
\usepackage{setspace}
\usepackage{epsf}
\usepackage[margin=10pt,labelfont=bf,labelsep=endash]{caption}
\usepackage{subcaption}
\usepackage{floatrow}
\usepackage{fullpage}
\usepackage[super]{nth}
\usepackage{enumitem}
\setlist[enumerate]{itemsep=1pt, topsep=0pt}
\setlist[itemize]{itemsep=1pt, topsep=0pt}
\usepackage[usenames,dvipsnames]{xcolor}
\usepackage{titlesec}
\usepackage[titletoc]{appendix}
\usepackage[nottoc,notlot,notlof]{tocbibind}
\usepackage[colorlinks,pagebackref=false]{hyperref}

\definecolor{dark-red}{rgb}{0.60,0.15,0.15}
\definecolor{dark-blue}{rgb}{0.15,0.15,0.65}
\definecolor{light-blue}{rgb}{.2,1,1}
\definecolor{medium-blue}{rgb}{0,0,0.5}
\definecolor{dark-green}{rgb}{0,.5,0}
\definecolor{indigo}{rgb}{0.29, 0.0, 0.51}

\hypersetup{
    pdftitle={Test-Optional Admissions},
    pdfauthor={Dessein and Frankel and Kartik},
    citecolor={dark-blue},
    bookmarksnumbered=true,
    urlcolor=indigo,
    linkcolor=dark-red}

\AtBeginDocument{} 

\makeatletter
  \renewcommand\@seccntformat[1]{\csname the#1\endcsname.{\hskip.7em\relax}} 
\makeatother

\newcommand{\appendixref}[1]{\hyperref[#1]{Appendix \ref{#1}}}
\newcommand{\onlineappendixref}[1]{\hyperref[#1]{Supplementary Appendix \ref{#1}}}

\newtheorem{lemma}{Lemma}

\newtheorem{proposition}{Proposition}

\newtheorem{corollary}{Corollary}

\theoremstyle{remark}

\theoremstyle{definition}
\newtheorem{example}{Example}

\titlespacing\section{0pt}{10pt plus 2pt minus 2pt}{4pt plus 2pt minus 2pt} 
\titlespacing\subsection{0pt}{6pt plus 2pt minus 2pt}{2pt plus 2pt minus 2pt} 
\titlespacing\subsubsection{0pt}{6pt plus 2pt minus 2pt}{0pt plus 2pt minus 2pt} 

\newcommand{\ind}{{1}\hspace{-2.5mm}{1}}
\newcommand{\Col}{^c}
\newcommand{\Soc}{^s}

\newcommand{\E}{\mathbb{E}}
\renewcommand{\epsilon}{\varepsilon}
\renewcommand{\phi}{\varphi}
\renewcommand{\bar}{\overline}
\DeclareMathOperator*{\supp}{Supp}

\def\u{\underline}
\def\o{\overline}
\newcommand{\Reals}{\mathbb R}

\newcommand{\usoc}{u^s}

\newcommand{\mailto}[1]{\href{mailto:#1}{\texttt{#1}}} 
\def\bi{\begin{itemize}}
\def\ei{\end{itemize}}
\newcommand{\citepos}[1]{\citeauthor{#1}'s (\citeyear{#1})}

\let\oldfootnote\footnote
\renewcommand\footnote[1]{\oldfootnote{\hspace{.4mm}#1}}

\makeatletter
\renewenvironment{proof}[1][\proofname] {\par\pushQED{\qed}\normalfont\topsep6\p@\@plus6\p@\relax\trivlist\item[\hskip\labelsep\bfseries#1\@addpunct{.}]\ignorespaces}{\popQED\endtrivlist\@endpefalse}
\makeatother

\begin{document}

\title{Test-Optional Admissions\thanks{We thank Nageeb Ali, Peter Arcidiacono, Chris Avery, Joshua Goodman, Michelle Jiang, Adam Kapor, Frances Lee, Annie Liang, Brian McManus, Victoria Mooers, Jacopo Perego, Marcin P\c{e}ski, Jonah Rockoff, Miguel Urquiola, Congyi Zhou, Seth Zimmerman, and various seminar and conference audiences for helpful comments and discussions. Guo Cheng, Tianhao Liu, and Yangfan Zhou provided excellent research assistance.}}

\author{Wouter Dessein\thanks{Columbia University, Graduate School of Business. Email: \mailto{wd2179@columbia.edu}.} \and Alex Frankel\thanks{University of Chicago, Booth School of Business; Email: \mailto{afrankel@chicagobooth.edu}.} \and Navin Kartik\thanks{Columbia University, Department of Economics; Email: \mailto{nkartik@columbia.edu}.}}
\maketitle
\thispagestyle{empty}

\begin{abstract}
 Many U.S. colleges now use test-optional admissions. A frequent claim is that by not seeing standardized test scores, a college can admit a student body it prefers, say with more diversity. But how can observing less information improve decisions? This paper proposes that test-optional policies are a response to social pressure on admission decisions. We model a college that bears disutility when it makes admission decisions that ``society'' dislikes. Going test optional allows the college to reduce its ``disagreement cost''. We analyze how missing scores are imputed and the consequences for the college, students, and society.
 \end{abstract}

\newpage
\setcounter{page}{1}
\setcounter{section}{0}
\setcounter{proposition}{0}

\section{Introduction}
\label{sec:intro}

With college admissions in the United States under increasing scrutiny, there is a vibrant debate about the role of standardized test scores.  The last decade has seen an increase in colleges going \emph{test optional}, i.e., not requiring applicants to submit standardized test scores. The University of Chicago \href{https://www.washingtonpost.com/local/education/a-shake-up-in-elite-admissions-u-chicago-drops-satact-testing-requirement/2018/06/13/442a5e14-6efd-11e8-bd50-b80389a4e569_story.html}{made waves} when it adopted this policy in 2018. By 2019, one third of the 900+ colleges that accepted the Common Application \href{ https://www.diverseeducation.com/students/article/15114774/report-testoptional-trends-in-pandemic-year-for-college-admissions}{did not require test scores}.

For obvious reasons, the Covid-19 pandemic dramatically increased the adoption of test-optional policies: in the 2021--22 application season, 95\% of Common-Application colleges \href{https://www.highereddive.com/news/share-of-common-app-colleges-requiring-admissions-tests-continues-to-plumme/610629/}{did not require test scores}. But even after the pandemic, notwithstanding a few prominent reversals, most U.S. colleges \href{https://www.highereddive.com/news/over-1700-colleges-wont-require-sat-act-for-fall-2023-up-from-same-poin/628267/}{have stayed test optional}. Furthermore, although our paper emphasizes college admissions, the shift away from requiring standardized tests 
is also pervasive in other segments of education.\footnote{According to \href{https://www.forbes.com/sites/akilbello/2022/01/14/did-harvard-just-signal-the-end-of-the-testing-era-it-started/?sh=645f36ff55b1}{Forbes magazine} in January 2022, ``The most public break-up [with standardized tests] has been in undergraduate admissions and the SAT/ACT, but kindergarten, high school, and graduate school admission offices have also been rejecting standardized tests \ldots [there is a] near-universal shift away from standardized tests that started before the pandemic but has accelerated in the last eighteen months.''}

Proponents of test-optional admissions often cite concerns that standardized testing may disadvantage low-income students and students of color. Indeed, many schools that go test optional claim to do so in order to increase the racial and income diversity on campus.\footnote{For example, in 2015 a George Washington University \href{https://gwtoday.gwu.edu/standardized-test-scores-will-be-optional-gw-applicants}{school official said that} ``The test-optional policy should strengthen and diversify an already outstanding applicant pool and will broaden access for those high-achieving students who have historically been underrepresented at selective colleges and universities, including students of color, first-generation students and students from low-income households”.}  But a test-mandatory college is free to admit students with low test scores if they are strong on other dimensions. Only with the 2023 U.S.~Supreme Court's ruling on \emph{SFFA v.~Harvard} has the consideration of race been banned. Moreover, other components of applications may also be subject to racial and income disparities,\footnote{In a 2016 \href{https://www.washingtonpost.com/news/grade-point/wp/2016/03/03/letters-of-recommendation-an-unfair-part-of-college-admissions/}{Washington Post opinion} titled `Letters of recommendation: An unfair part of college admissions,' John Boeckenstedt from DePaul University argues that: ``If you wanted to ensure that kids from more privileged backgrounds have a better chance to get into the schools with the most resources, letters of recommendation would be one of the things you’d start with.''} and test scores are unlikely to be completely uninformative about college preparedness. Indeed, MIT reinstated its testing requirement for the 2022-23 admissions cycle, \href{https://mitadmissions.org/blogs/entry/we-are-reinstating-our-sat-act-requirement-for-future-admissions-cycles/}{arguing that} ``standardized tests help us identify socioeconomically disadvantaged students who lack access to advanced coursework or other enrichment opportunities
that would otherwise demonstrate their readiness for MIT.''  Similarly, a  \href{https://senate.universityofcalifornia.edu/_files/committees/sttf/sttf-report.pdf}{2020 report} by the University of California found that standardized test scores help predict student success, across demographic groups and disciplines, even after controlling for high school GPA \citep{UCreport2020}.

Hence a puzzle: if a college can use test scores as it would like, why choose not to have access to a student's score? Why throw away potentially valuable information? Indeed, in \citet{DFK24-PP}, we show that under a broad set of conditions, a college that can freely use information---and commit to how it will do so---cannot benefit from going test optional. The conditions allow for differential costs of test preparation and different distributions of test scores for reasons unrelated to ability.\footnote{The conditions do preclude prohibitive costs of sitting, as opposed to studying, for the test. It is not clear to us that, outside of pandemics, the costs of sitting for a test are a compelling rationale for going test optional. In fact, prior to the Covid-19 pandemic, \href{https://www.edweek.org/teaching-learning/which-states-require-students-to-take-the-sat-or-act}{25 U.S. states required} either the SAT or ACT for high school graduation.}

There are a few ways out of the puzzle. Students may not trust admission policy pronouncements. Being test optional could then be a credible signal that a college does not seek to put much weight on test scores. Relatedly, a college may be limited in how effectively it can control its admissions officers, who may put too much weight on test scores due to moral hazard or other reasons. An alternative explanation is that students may simply dislike taking tests, and hence apply to otherwise-similar colleges that do not require them. Yet another possibility is that test optional allows colleges to report higher average (submitted) standardized test scores and thereby improve their rankings; \citet{CDCC13} provide empirical evidence of  strategic admission decisions in this regard.

Our paper proposes, instead, that test-optional policies are driven by social pressure on admission decisions. 
When a selective college admits a low-scoring student while rejecting a high-scoring student with an otherwise similar GPA, it may be subject to social pressure from a community that disagrees with the weight it puts on tests versus legacy status or racial diversity. Such pressure may come from the broader US public; indeed, polls show that the general public believes race and legacy status should be deemphasized relative to test scores.\footnote{In a \href{https://www.pewresearch.org/fact-tank/2022/04/26/u-s-public-continues-to-view-grades-test-scores-as-top-factors-in-college-admissions/}{2022 PEW research survey}, only 26\% of respondents thought that race or ethnicity should be even a minor factor in college admissions, with 25\% for legacy status. By contrast, 39\% thought that test scores should be a major factor, and an additional 46\% thought they should be a minor factor.} Public pressure is exemplified by the  lawsuits that led to the aforementioned 2023 U.S.~Supreme Court ruling against affirmative action in college admissions. 
Alternatively, social pressure may come from a college's internal stakeholders with the opposite views; for instance, current students may have a stronger preference for diversity in admissions than the college administration.

Our broad argument is that by hiding score disparities among students who do not submit their test scores, a college lowers the cost of disagreement with ``society''. The lower disagreement cost may also lead the college to admit students it likes more, based on diversity, extracurriculars, or legacy preferences---but this is not necessary for the college to benefit from not seeing test scores. Notably, our argument does not rely on any naivety: we assume that society is Bayesian and understands that students who don't submit scores tend to have lower scores. Also important, we show that being test optional can help a college regardless of whether, for any given group of students, it wishes to be less selective than society (i.e., to use a lower test-score threshold) or  more selective (a higher threshold).  

In more detail, our model in \autoref{sec:model} has a college with preferences over which students to admit, based on both their non-test observable characteristics (e.g., GPA, race, SES, extracurriculars, and legacy status)~and test scores. Society has its own preferences. Society does not make any decisions, but the college
places some value on minimizing disagreement between its admission decisions and those that society would make.  The college commits to an admissions policy: an acceptance rule mapping observables and test scores into an admission decision, and, in a test-optional regime, an \emph{imputed} test score that it assigns to students who don't submit scores (as a function of non-test observables). A student submits their test score if and only if it is higher than the score the college would impute. Society assesses test scores in a Bayesian manner:
non-submitters are evaluated based on their expected test score, given non-test observables and submission behavior. 

Whenever society disagrees with the college's admission decision, the college incurs a \emph{disagreement cost}. Importantly, this cost is only based on the admission decision given the available information; society does not judge the college for its choice of information.\footnote{In their experimental study concerning the use of test scores, \citet{LiangXu24} emphasize direct preferences over what information to observe.}
If the college accepts an applicant that society wants to reject, the disagreement cost is proportional to society's disutility from acceptance. If the college rejects an applicant society wants to accept, this cost is proportional to society's disutility from rejection.
The college chooses its admissions policy---both the imputation and acceptance rules---to maximize its ex-ante expected utility from admissions decisions less disagreement costs.  

When a college can freely choose its imputation rule, the college can't be worse off under test optional than test mandatory. It could simply replicate the test-mandatory outcome by imputing a low enough test score that all students submit. Our key insight, though, is that the college can benefit---strictly---from going test optional. 

To see how, consider the case of a student with non-test observables such that the college is less selective than society: the college has a lower test-score bar than society to admit this type of applicant. For instance, take students who excel in fencing and suppose the college values able fencers more than society.\footnote{According to the \href{https://www.nytimes.com/2022/10/17/us/fencing-ivy-league-college-admissions.html?smid=url-share}{New York Times} in October 2022, ``a way with the sword can help students stand out in the college admissions game\ldots because each good school, especially Ivy League schools, have fencing.''}  
One option for the college is to impute a very high test score for 
fencers, with the policy of admitting all those with the imputed score (or higher). Then none of the 
fencers submit their scores, and all of them are admitted. The cost for the college is that it admits some very low-scoring fencers. The benefit, though, is that bringing high-scoring fencers
into the non-submission pool reduces disagreement costs from admitting some fencers that the college wanted but society did not. Indeed, if society is willing to accept 
fencers with average test scores, then imputing a very high score allows the college to accept all of these now-undifferentiated fencers at \emph{zero} disagreement cost. 
At the extreme, if the college prefers to admit every 
fencer regardless of test score, it obtains its first best for this group---they are all admitted, with no disagreement cost.

Now consider students with observable characteristics at which the college is \emph{more} selective than society. Suppose the college prefers to admit applicants from New Jersey only if they score above 55, whereas society loves the Garden State and would like to admit any of its students with a score above 25. If test scores are submitted, the college incurs a disagreement cost for any rejected applicant with a score above 25. Consequently, under test mandatory, the college uses a score threshold between 25 and 55, say 40. Under test optional, however, the college can do strictly better among New Jerseyans by imputing a score between 40 and 55 and then rejecting non-submitters. Imputing the score of 40 would replicate the test-mandatory admissions outcome but lower the disagreement cost because all New Jerseyans with scores below 40 don't submit; now there is no differentiation between those below 25, where there is no disagreement, and those in the 25--40 range, where there is disagreement. The college may do even better by imputing a score strictly above 40, which would reject more students and thus improve, from its perspective, its New Jerseyan student body.

We show in \autoref{sec:analysis} that the above examples encapsulate the general logic for how a college can benefit from going test optional. Notice that in these examples, fencers  benefit---some weakly and some strictly---from a school going test optional, whereas New Jerseyans are hurt. \autoref{sec:flexible} establishes that these consequences for student welfare hold generally: student groups for whom the college is less selective than society benefit from test optional, while student groups for whom the college is more selective are hurt. Furthermore, at every student observable, society is (weakly) harmed by test optional.

For test optional to never harm a college, the imputation rule must be judiciously chosen. In practice, 
many schools promise that non-submitters will be treated ``fairly''. The \href{https://admission.usc.edu/apply/test-optional-policy-faq/}{University of Southern California's statement} is representative:  ``applicants will not be penalized or put at a disadvantage if they choose not to submit SAT or ACT scores.'' While the meaning of such policies is ambiguous, one interpretation is that of
a \emph{no adverse inference} imputation rule: a student who does not submit a test score is imputed their expected score given other observables, but crucially, not conditioning on non-submission.  \autoref{sec:restricted} studies test-optional outcomes under this or some other fixed imputation rule. We establish a sense in which students with good non-test observables (and low test scores) benefit when a college goes test optional because it increases their admission rate. Students with intermediate observables (and intermediate scores) are harmed. Other students are unaffected.

When constrained to use an imputation rule like no adverse inference, colleges may be worse off under test optional than test mandatory (by contrast with flexible imputation). Determining the college's preferred testing regime requires more structure on the environment.  We show that, under mild conditions, for any fixed imputation rule, the college is worse off under test optional if either the social-pressure intensity is low or the college's preferences are similar to society's. We use this result to discuss the recent reversal by a few colleges---e.g., \href{https://president.dartmouth.edu/news/2024/02/reactivating-satact-requirement-dartmouth-undergraduate-admissions}{Dartmouth} and \href{https://hub.jhu.edu/2024/08/16/undergrad-admissions-standardized-test-requirement/}{John Hopkins}---to test mandatory.

Towards further insight on restricted imputation, we also discuss an extended example in \autoref{sec:AA}. There, we study how the college's choice of testing regime can depend on whether affirmative action---that is, an admission rule that directly conditions on race---is allowed. Our interest stems from the public and legal debate around the use of affirmative action in the US, culminating in the 2023 Supreme Court ruling that severely limited race-based admissions. In the example, we show that banning affirmative action can push a college from preferring test mandatory to preferring {test blind}.\footnote{Test blind is when students simply cannot submit tests scores, or the college ignores test scores entirely. In our model, this is equivalent to test optional in which non-submission is imputed as the highest test score.} The intuition is that if students in the college's favored group have lower test scores, then the college values tests less when it cannot condition on group membership. We explain how a ban on affirmative action may thus ``backfire'' on society, which prefers that the college be test mandatory.

Finally, \autoref{sec:competition} briefly explores competition between colleges. A test-optional college facing a test-mandatory competitor may face a form of adverse selection: if the college admits students who do not submit test scores, the students with unobserved low scores are more likely to accept this admission offer, because they are less likely to be admitted somewhere else. This adverse selection may push a college to match its competitor's testing regime, choosing test-optional only when the other is test-optional. As we show, however, this force may also go in the opposite direction, meaning that a college sometimes wants to mismatch its competitor's testing regime.

\paragraph{Related literature.} There are several empirical papers studying test-optional (or test-blind) college admissions using data from prior to the Covid-19 pandemic \citep[e.g.,][]{belasco2015test,saboe2019sat,bennett2022untested}. In a  review, \citet[pp.~53--54]{dynarski2022survey} conclude that test-optional policies had limited effect on increasing diversity and applications, but may have helped colleges boost their public rankings by raising the average (submitted) standardized test score of enrolled students.
Using data from a sample of student test-takers in the 2021-22 admission cycle, \citet{McManus23} document sophisticated submission behavior. Not only did students withhold low scores, but they conditioned their choice on their other academic characteristics as well as colleges’ selectivity and testing policy statements.

The use of standardized tests in college admissions has been studied in economic theory as well. \citet{krishna2022pareto} propose pooling test scores into coarse categories to reduce the wasteful costs of test preparation. \citet{lee2021gaming} study how low-powered selection---such as putting less weight on test scores---may help a college by reducing students' incentives to improve their scores.\footnote{More broadly, in a ``muddled information'' framework \citep{FK19}, \citet{FK22} and \citet{ball2022scoring} explore how a decisionmaker should commit to underutilize manipulable information to improve decision accuracy.} \citet{GargLiMonachou20} assume that some students have no access to standardized tests, which means that a test-optional/blind policy broadens the applicant pool even though it provides less information about those who do apply. \citepos{borghesan2022heterogeneous} structural analysis of college admissions also emphasizes students' costs of taking standardized tests: going test blind reduces a college's information but allows students with high test-taking costs to apply. He predicts that this policy would reduce student quality at top schools without increasing diversity.
Related to costly test-taking is \citepos{ottaviani2020grantmaking} model of (grant) allocation with costly application. They show that using more noisy measures of applicant quality can enlarge an applicant pool.

In contrast to the papers in the preceding paragraph, our argument for why colleges benefit from going test optional does not rely on the cost of obtaining or improving test scores, nor on the cost of applying to a college. Our model assumes that students are simply endowed with a test score and application is costless.\footnote{In their empirical studies, \citet{goodman2016learning} and \citet{hyman2017act}) find that government policies mandating high school students to take standardized tests increase college enrollment rates of low-income students, either because the students discover they are higher-achieving than they thought or because colleges discover and then recruit students through such testing. More generally, scholars have suggested that eliminating application barriers for low-income students can increase the number of students that apply to and enroll in selective colleges \citep{hoxby2012missing,hoxby2013expanding,goodman2020take}.}

Among theoretical papers on affirmative action in college admissions, a topic we take up in \autoref{sec:AA}, most related is \citet{chan2003does}.\footnote{Various other papers model aspects of college admissions that we do not address, such as early admissions \citep[e.g.,][]{avery2010early}, managing enrollment uncertainty \citep[e.g.,][]{che2016decentralized}, college tuition determination \citep[e.g.,][]{fu2014equilibrium}, and which colleges a student should apply to \citep[e.g.,][]{ChadeSmith06,alishorrer2023}.} They  model a college that values both student quality and diversity. When affirmative action is banned, the college may adopt an admission rule that puts less weight on academic qualifications, such as standardized test scores,  in order to promote diversity. The logic is related to that of statistical discrimination \citep{phelps1972statistical,arrow2015theory}, except that instead of race serving as a signal of qualification, qualification serves as a signal of race. 
Notably, \citet{chan2003does} do not provide a rationale for why a college strictly benefits from not observing test scores; in their model, being test blind is equivalent to being test mandatory and putting zero weight on tests. In our model, social pressure can lead a college to strictly prefer test blind.

Our paper also connects to the large literature on voluntary disclosure of verifiable information. The canonical result here is that of ``unraveling'' \citep{Grossman81,Milgrom81}, which corresponds to all students submitting their scores even when it is optional.  \href{https://thehill.com/changing-america/enrichment/education/3758713-in-college-admissions-test-optional-is-the-new-normal/}{It is reported}, however, that fewer than half of U.S. college applicants who applied early decision in Fall 2022 submitted test scores. Unraveling does not arise in our model because the college can commit to how it will treat students who do and do not submit their score.

Finally, in our model, the college's and society's information depends on which students submit test scores. This is determined by the testing regime and, under test optional, the college's imputation rule. Our work thus relates to Bayesian persuasion and information design \citep{KG11,bergemann2019information}. Unlike in much of that literature, our college cannot choose arbitrary information structures.
One paper on information design that connects to our themes is 
\citet{LiangLuMuOkumura24}. They explore how a designer may ban the use of certain inputs, such as test scores, because of disagreement with how a decisionmaker would use those inputs. While they explain why society may choose to prevent a college from using test scores, we show why a college may itself choose to not see test scores. Like us, \citet{LiangLuMuOkumura24} also discuss why a college may choose to not see test scores if society bans affirmative action; however, unlike our mechanism of social pressure, theirs entails conflicting preferences between the college and its admissions officers.

\section{An Illustrative Example} \label{sec:motivatingex}

Consider a single student who has applied to a college. (An alternative interpretation is that of a mass of students who share common observable characteristics.) The student's test score $t$ is drawn from a uniform distribution between 0 and 100. Society's utility from admitting the student is $\usoc(t) = t-40$, and its utility from not admitting the student is normalized to $0$. So, ignoring indifference, society wants to admit the student if and only if their test score is above 40. 
The college receives some information about the student's test score---we will consider different possibilities below---and then chooses whether to accept or reject the student. Society then judges the college's decisions given the available information. Importantly, the college and society have the same information; information asymmetry between them is not our driving mechanism. Rather, what is crucial is that the college faces disagreement costs from social pressure for making decisions that society disagrees with.

\paragraph{Disagreement cost.}  The disagreement cost is proportional to the extent of society's disagreement with the college's decision, given the available information. Concretely, disagreement equals the increase in society's expected utility if society were to make admission decisions as opposed to the college. If the college accepts the student and society would also prefer to accept them (i.e., $\mathbb{E}[\usoc(t)] > 0)$, or the college rejects the student and society would also prefer to reject ($\mathbb{E}[\usoc(t)] < 0$), then the college bears no disagreement cost. That is, in each of those cases, the respective disagreement costs $d_{A=1}$ and $d_{A=0}$ are both $0$, where $A=1$ denotes acceptance and $A=0$ denotes rejection. However, if the college rejects the student when society prefers to accept, the college bears a disagreement cost of $d_{A=0}=\E[\usoc(t)]>0$. Likewise, if the college accepts a student that society prefers to reject, the disagreement cost is $d_{A=1}=-\E[\usoc(t)]>0$.   See \autoref{fig:example40}.

\paragraph{Why not observe test scores?}We now illustrate how the college can reduce disagreement costs by not observing test scores. 

\begin{figure} \caption{Disagreement cost from accepting ($A=1$) and rejecting ($A=0$) an student.\label{fig:example40}}
	\[\includegraphics[width = 5.5 in]{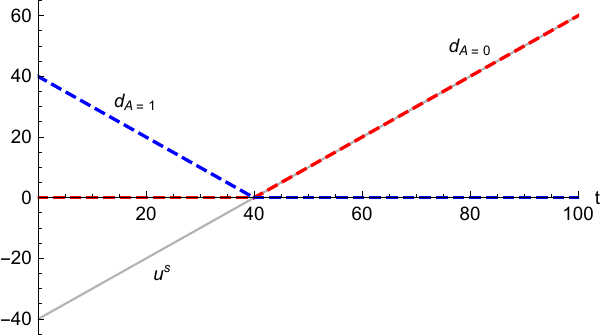}\]
\end{figure}

First consider \emph{test mandatory}: the student's test score is observed. If the college chooses to accept regardless of the test score, it bears a disagreement cost of $40-t$ whenever the score is below $40$ (and $0$ otherwise), and so the expected disagreement cost is $\int_0^{40}\frac{1}{100} (40-t)\mathrm{d} t = 8$. Analogously, if the college instead chooses to reject regardless of test score, it bears an expected disagreement cost of $\int_{40}^{100} \frac{1}{100} (t-40)\mathrm{d}t = 18$.

Now consider \emph{test blind}: the student's test score is not observed. Here, with no information beyond the uniform prior over the test score, society evaluates the student as if their score were equal to the expected value $\E[t]=50$. If the college accepts the student, it now faces a disagreement cost of 0: absent test score information, society agrees that the student should be accepted. So if the college were going to accept the student regardless of their test score, then hiding the score reduces its expected disagreement cost from 8 to 0.

If the test-blind college rejects the student,  it does face a disagreement cost: society's expected utility from admitting the student is $\E[t]-40 = 10$, and so the college's disagreement cost from rejection is 10. Nonetheless, hiding the test score reduces the expected disagreement cost of rejecting all applicants from $18$ to $10$.

The upshot is that for either decision the college makes---so long as it is independent of the test score when that is observed---the college can reduce expected disagreement cost by hiding the test score, i.e., going test blind. The fundamental reason is that both disagreement cost curves $d_{A=1}(t) = \max\{40-t, 0\}$ and $d_{A=0}(t) = \max\{t-40, 0\}$ are convex, as seen in \autoref{fig:example40}. Mathematically, the reduction of expected disagreement cost by going test blind is a consequence of Jensen's inequality: the disagreement cost evaluated at the expected test score is smaller than the expectation of the  disagreement costs across test scores.

We highlight two points. First, social pressure induces a form of non-consequentialist preferences: the college has preferences over what information it has even when its decision is independent of the test score. Second, society judges the college only on its admission decision given the available information. 
That is, the disagreement cost is evaluated using the information the college has, not the information it could have had.\footnote{Mathematically, if the disagreement cost from admitting a student without a test score were $\E\left[\max\{-u^s(t),0\}\right]$ instead of our assumed $\max\{-\E[u^s(t)],0\}$, then admitting the student under test blind would lead to the same disagreement cost as under test mandatory.} 
Moreover, the college is not directly penalized for its choice not to require test scores.

\paragraph{Test optional.} If the college seeks to admit
only some students---rather than accepting or rejecting all of them---it might improve upon test blind by going \textit{test optional}: the student can choose whether to submit their score.  For example, suppose the college wants to admit students with test scores above 60, while society's preferred threshold remains 40. Under test optional, the college could commit to treat non-submitters ``as if'' they have a score of 60, and only accept students with scores (strictly) above 60. 
If students with scores below 60 (optimally) do not submit their score,\footnote{A small acceptance probability for students with a score of 60, including non-submitters, would make this strategy strictly optimal for students.} this policy implements the college's desired threshold with zero disagreement cost. (Society's expected utility from admitting a non-submitting student is  $\E[t|\text{non-submission}]-40 = -10$, so it agrees with rejecting non-submitters.)

\paragraph{ A tradeoff.} In general, a college faces a tradeoff between using information to make better decisions and not seeing information to reduce disagreement costs. We explore this tradeoff in the rest of the paper. We study how test-optional colleges decide which applicants to admit, how students choose whether to submit test scores, and how the resulting outcomes differ from a test-mandatory benchmark.

\section{A Model of Admissions under Social Pressure}
\label{sec:model}

We model a student applying to a college, with a broader ``society''  playing a passive role.  The student can be viewed as a representative applicant; we will sometimes use the plural students for exposition. Society represents any external group that might scrutinize admission decisions and has preferences over who ought to be admitted: alumni, parents, local governments, the popular press, and even the judicial branch.

The student is endowed with some publicly observable characteristics and a test score, which is their private information. In a test-mandatory regime, the student mechanically submits their test score, making it public to the college and society. In a test-optional regime, the student chooses whether to submit their score.  In either regime, the college chooses whether to admit the student based on their observable characteristics and, if submitted, their test score. Both the college and society have preferences over whether the student should be admitted as a function of their observables and their true test score. The college also places some weight on reducing disagreement between its admission decision and the decision society would want it to make, given all available information. 

\subsection{Model Primitives}
\label{sec:primitives}

\paragraph{Observables and test scores.} Formally,  the student/applicant has a \emph{type} $(x,t)\in \mathcal X \times \Reals$, where $x$ is an \emph{observable} (or vector of \emph{observables}) and $t$ is the \emph{test score}. The distribution of observables is given by $F_x$ and the test score has conditional distribution $F_{t|x}$.\footnote{More precisely, $\mathcal X$ is a measurable space and $F_x$ is a probability measure on that space. To simplify some technicalities, we assume that for each $x$, $F_{t|x}$ is either continuous or is discrete with no accumulation points, and that all relevant expectations exist. \label{fn:continuousdiscrete}} 

The observable $x$ is public information to all players. The test score $t$ is private information to the student, which may be submitted ($S=1$) or not ($S=0$). Submitting the score makes it observable to all other players.  Our primary interest is in two college admission regimes: \emph{test mandatory}, in which test scores must be submitted, and \emph{test optional}, in which scores may be submitted. We will also talk about \emph{test blind}, wherein the score cannot be submitted.

\paragraph{Preferences.} The college decides whether to admit the student (denoted $A=1$) or not ($A=0$), based on observables $x$ and, if submitted, the test score $t$. The student strictly prefers a higher probability of being admitted. Society's utility and the college's material or ``underlying'' utility if the student is accepted are given, respectively, by
\begin{align*}
    u^s(x,t)&:=v^s(x)+w^s(x) t,\\
    u^c(x,t)&:=v^c(x)+w^c(x) t,
\end{align*}
where the superscripts have the obvious mnemomic (\emph{s}ociety and \emph{c}ollege), and each $w^i(\cdot)>0$ for $i=s,c$. We view monotonicity of these preferences in the test score as natural; the affine specifications aid subsequent interpretation and tractability. Both society's and the college's underlying utility are normalized to $0$ if the student is not admitted.

In addition to its underlying utility, the college suffers disutility from social pressure on its admission decision. To formalize that disutility, let $t^s$ denote the test score society treats the student as having; this will be determined endogenously. 
Anticipating equilibrium, think of $t^s=t$ if the score is submitted, and $t^s=\E[t|x,S=0]$ under non-submission. For any $t^s$, society's \emph{disagreement} with the college's decision is given by
\begin{equation}
\label{eq:d}
d(x,t^s,A):=\begin{cases}
\max\{u^s(x,t^s),0\} & \text{if } A=0,\\
\max\{-u^s(x,t^s),0\} & \text{if } A=1.
\end{cases}
\end{equation}
This disagreement captures society's benefit if it were to decide on the admission decision instead of the college. There is no disagreement if, given the available information, society's preferred decision is the same as the college's decision; but when there is a conflict in preferred decisions, then disagreement is linear in the magnitude of society's expected benefit from its preferred decision. As before, the monotonicity here is natural; linearity is for tractability. The assumed linearity of $u^s(x,t)$ in the test score $t$ is what allows \eqref{eq:d} to be written with $u^s(x,t^s)$ instead of $\E[u^s(x,t)]$, where the expectation is with respect to the distribution of $t$ given $x$ and the student's (non-)submission.

The college's overall payoff $U^c$ is its underlying utility less the (scaled) disagreement:
\begin{equation}
\label{eq:U^c}
U^c(x,t,t^s,A):= 
A u^c(x,t)-\delta d(x,t^s,A),
\end{equation}
where $\delta>0$ is a parameter capturing the intensity of social pressure on the college. We refer to $\delta d (\cdot)$ as the \emph{disagreement cost} to the college.

\paragraph{Admissions policies.}
The college's admissions policy has two components, one of which---how to treat students who don't submit test scores---is irrelevant under test mandatory.  

First, given the student's observable $x$, we assume that the college treats non-submission of a test score as equivalent to some specific test score, which we call the \emph{imputation}. More precisely, there is an \emph{imputation rule} $\tau:\mathcal X\to [-\infty,+\infty]$,\footnote{The co-domain is the extended reals for technical convenience when test scores can be arbitrarily small or large; if test scores lie in a compact set, then we could take the co-domain of $\tau$ to be that compact set.} with $\tau(x)$ the imputation for observable $x$. We will be interested in two settings: either the college can choose the imputation rule arbitrarily, which we call \emph{flexible imputation}, or the imputation rule is exogenously given, which we call \emph{restricted imputation}. 

Second, the college chooses an \emph{acceptance rule} $\alpha:\mathcal X \times [-\infty,+\infty] \to [0,1]$, where $\alpha(x,\hat t)$ is the probability of admitting a student with observable $x$ and imputed/submitted test score $\hat t$. We stress that the acceptance rule cannot (directly) condition on the student's true test score, and it does not distinguish between imputed and submitted scores---this captures our notion that imputing a score means treating a non-submitting student as if they have  submitted that imputed score. As in \citet{chan2003does}, we assume that $\alpha$ must be \emph{monotonic} in the sense that for any $x$, $\alpha(x,\cdot)$ is weakly increasing.

\paragraph{College's problem.} Since the college's acceptance rule is monotonic, there is a simple best response for the student: submit their score if $t>\tau(x)$ and don't submit if $t\leq \tau(x)$. We restrict attention to the student playing this strategy. Given this student strategy, we assume society is Bayesian in evaluating the student. In particular, if the student submits their  test score, then $t^s=t$; if the student does not submit, then $t^s=L(\tau(x)| x)$, where $L$ (mnemonic for ``lower expectation'') is defined by 
$$L(t'|x):= \E[t | t\leq t', x].\footnote{For $t' \leq \inf \supp[ F_{t|x}]$, we set $L(t'|x):=\inf \supp[F_{t|x}]$.}$$

The college's problem is to choose---commit to---its imputation rule $\tau$ (under test optional with flexible imputation) and its acceptance rule $\alpha$, to maximize its expected payoff $U^c$, anticipating the student's best response and society's Bayesian inferences. Note that when $\tau(x)=\infty$, the college is effectively test blind among students with observable $x$, and analogously it is effectively test mandatory when $\tau(x)=-\infty$.

\subsection{Discussion of the Model}

\subsubsection{Imputation and acceptance rules}

A test-optional admissions policy in our model is an imputation rule paired with a monotonic acceptance rule. We view the framework of imputation as an appealing and versatile way to capture how colleges may actually treat missing test scores. For example, it allows us to discuss cultural or legal norms about how non-submitters should be treated (as elaborated below).  Monotonicity of the acceptance rule is without loss if students can ``freely dispose'' of test scores---a student with test score $t$ can costlessly reduce it to any value less than $t$. 

At a theoretical level, however, the natural alternative would be to specify an admissions policy as an arbitrary mapping from observables, whether the student submits their score, and the score if submitted, to an admissions probability. We show in \autoref{sec:general_policies} that the outcome under this alternative is the same as that under flexible imputation. In other words, given flexible imputation, it is without loss of generality to stipulate that the college treats missing test scores via imputation and uses a monotonic acceptance rule.

\subsubsection{Restricted imputation rules}
\label{sec:model-restrictedimputation}

With flexible imputation, the college can arbitrarily choose how to impute missing test scores. With restricted imputation, we consider the other extreme, in which an imputation rule is exogenously specified. 
Although our analysis will not have any results tied to particular restricted imputation rules, we allow for them to cover some colleges' practice of publicly promising not to ``penalize'' or ``disadvantage'' students who don't submit scores. We interpret such promises as mapping to some version of what we call the \emph{no adverse inference} imputation rule, $\tau(x) = \mathbb{E}[t|x]$. Contrast this expression to the Bayesian imputation rule used by society, in which $t^s = \mathbb{E}[t|x, S=0]$: no adverse inference updates based on observables but not on the choice not to submit.  That is, the college imputes test scores as if students who did not submit chose to do so non-strategically.\footnote{After switching to test optional in 2020, Dartmouth \href{https://web.archive.org/web/20220706131809/https://admissions.dartmouth.edu/follow/blog/lee-coffin/dartmouth-adopts-test-optional-policy-class-2025}{announced} ``Our admission committee will review each candidacy without second-guessing the omission or presence of a testing element.''}

Even when ignoring the submission decision, the college might condition its expectation only on some subset of observables.  For instance, if the observable vector $x = (x_0, x_1)$ has component $x_1$ corresponding to ``grades'' and component $x_0$ corresponding to ``demographics'', the college might impute $\tau(x) = \mathbb{E}[t|x_1]$ rather than $\tau(x)=\mathbb{E}[t|x]$. Indeed, certain demographic features such as race are legally protected categories, and it may be forbidden to impute scores differently based on these factors---even if they are in fact predictive. In the limiting case, a college might deem \emph{all} observables irrelevant, in which case it would impute $\tau(x) = \mathbb{E}[t]$ identically for all applicants.

\subsubsection{Key assumptions}

\paragraph{Simplifications.} Our model makes a number of simplifying assumptions in order to focus on the channel of social pressure as an explanation for going test optional. For instance, we abstract away from a student's decision of how much to study for, or whether to even take, the test. Instead, we endow students with a test score. We then give the college and society a reduced form preference over these test scores rather than microfounding any inference over underlying ability. We also don't model the student's application decision. 

Another simplification is that our college has a fixed underlying utility threshold for admission. In particular, even if a switch from test mandatory to test optional leads to a different number of admitted students, the college does not raise or lower its  threshold for admission in order to keep its class size constant.\footnote{\label{fn:capacity}Our model can be consistent with a capacity constraint if we interpret the zero utility level as the utility both the college and society get from admitting students from a group in excess supply, with no social pressure. For example, many highly-selective colleges claim they could fill their entire class with a group of fairly homogeneous students who submit near-perfect SAT scores.} We return to this point in the \hyperref[sec:conclusion]{Conclusion}.

\paragraph{Student submission behavior.}

\cite{McManus23} show that student submission behavior does appear to vary with their belief about how colleges might impute missing test scores. We assume that students submit a test score if their true score $t$ is strictly above the college's imputed value $\tau(x)$, and they withhold the score if $t$ is weakly below $\tau(x)$. Higher test scores can only help admission chances. So, as discussed, this strategy guarantees a student the highest chance of admission. While there may be other optimal student strategies (when submitting a test score $t$ would lead to the same acceptance probability as not submitting),\footnote{In particular, when $t=\tau(x)$, the student is necessarily treated identically regardless of whether they submit; the behavior of these student types is immaterial if there are no mass points in the score distribution at $t=\tau(x)$.} a student can safely follow the strategy we focus on even if they do not know which (monotonic) acceptance rule the college is using.

We note that although the strategy is robust to a student's uncertainty over the college's acceptance rule, it is sensitive to the student's belief about their imputed test score. Of course, in the real world, students face uncertainty about how colleges treat missing test scores; we return to this point in the \hyperref[sec:conclusion]{Conclusion}.

\subsection{Ex-Post Utility} 
\label{sec:expost}
When $t^s=t$, as will be the case if the student submits their score, \autoref{eq:d} and \autoref{eq:U^c} imply that the college's net benefit from admitting the student is given by
\begin{align}
U^c(x,t,t,1)-U^c(x,t,t,0)&=u^c(x,t)-\delta\left[d(x,t,1)-d(x,t,0)\right] \notag\\
&=u^c(x,t)+\delta u^s(x,t) \notag\\
&\propto \frac{1}{1+\delta}u^c(x,t)+\frac{\delta}{1+\delta} u^s(x,t) \notag\\
&=:u^*(x,t).\label{eq:u^*} 
\end{align}
We refer to $u^*(x,t)$ as the college's \emph{ex-post utility}. For a score-submitting student, our disagreement cost formulation implies that the college's net benefit  from admission is equivalent to (i.e., proportional to) a convex combination of the college's underlying utility and society's utility. If the student submits their score, the college's payoff is maximized by admitting the student if and only if (modulo indifference) $u^*(x,t)>0$.

For $i\in\{c,s\}$, we refer to $\u t^i(x)$ such that $u^i(x,\u t^i(x))=0$ as the college/society's test-score \emph{bar} for admission: it is the score threshold such that each would---if unencumbered by social pressure---prefer to admit the student with observable $x$ if and only their score is above that threshold. We denote the ex-post utility bar by $\u t^*(x)$; it is defined by $u^*(x,\u t(x))=0$ and is the threshold above which, accounting for social pressure, the college wants to admit the student.\footnote{More explicitly, since $u^i(x,t)=v^i(x)+w^i(x) t$ and $u^*(x,t)=\left(u^c(x,t)+\delta u^s(x,t)\right)/(1+\delta)$, we compute $\u t^i(x)=-v^i(x)/w^i(x)$ and $\u t^*(x)=-\left(v^c(x)+\delta v^s(x)\right)/\left(w^c(x)+\delta w^s(x)\right)$.} We say that the college is \emph{less selective} than society at observable $x$  if $\u t^c(x)<\u t^s(x)$, while it is \emph{more selective} if $\u t^c(x)>\u t^s(x)$.  In either case, the ex-post utility bar $\u t^*(x)$ is in between the two parties' bars, and it monotonically shifts from $\u t^c(x)$ to $\u t^s(x)$ as the social-pressure intensity parameter $\delta$ increases.

\autoref{fig:prefs} illustrates with a leading specification in which $x \in \Reals$, and for each $i\in\{c,s\}$,  $u^i(x,t)=a^i+x+w^i \times t$. In this specification, the college weights test scores more than society when $w^c>w^s$, and weights test scores less than society when $w^c<w^s$. The three lines indicate the respective test-score bars  at each $x$. 
When the college weights test scores less, as in the figure's left panel, at low $x$ it is more selective (has a higher bar) than society, but at high $x$ it is less selective (has a lower bar); and the reverse when the college weights test scores more than society, as in the right panel.

\begin{figure}[h]
\caption{Test score admission bars for society ($\u t^s$), the college's underlying utility ($\u t^c$), and ex-post utility ($\u t^*$). For this figure, $x \in \Reals$
and $u^i=a^i+x+w^i \times t$.
\label{fig:upics}}

\begin{centering}

\begin{tabular}{cc}
\begin{subfigure}{0.45\textwidth}
\includegraphics[width = 2.75 in]{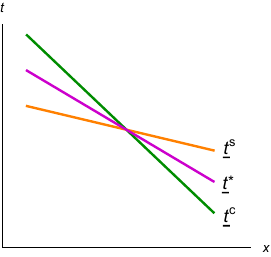}
\caption{College weights tests less than society: $w^c<w^s$.}
\end{subfigure}
&
\begin{subfigure}{0.45\textwidth}
\includegraphics[width = 2.75 in]{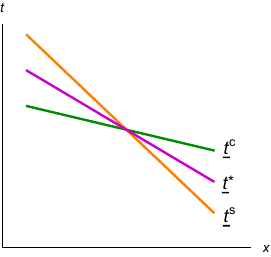}
\caption{College weights tests more than society: $w^c>w^s$.}
\end{subfigure}
\end{tabular}

\end{centering}
\label{fig:prefs}
\end{figure}

\subsection{Test-Mandatory Admissions}
\label{sec:mandatory}

In a test-mandatory regime, both the college and society always know the student's score. In light of social pressure, the college simply maximizes its ex-post utility for each $(x,t)$; its admission decision is determined by the ex-post bar.

\begin{proposition}
\label{prop:mandatory}
    Under test mandatory, the college admits a student with observable $x$ if $u^*(x,t) > 0$ (equivalently, $t > \u t^*(x)$) and rejects the student if $u^*(x,t) < 0$ (equivalently, $t < \u t^*(x)$).
\end{proposition}

As the social-pressure intensity parameter $\delta$ increases, the college becomes less selective at observable $x$ if, based on its underlying utility, it is more selective than society ($\u t^c(x)>\u t^s(x)$), and conversely if it is less selective than society. Plainly, the student with observable $x$ benefits in the former case and is harmed in the latter case.\footnote{Benefit/harm here is in the sense of set inclusion. For example, suppose the college is more selective than society at $x$. Then a student with that observable may be rejected when social pressure intensity is low, and admitted when intensity is high; or they may receive the same outcome at both intensities.
}

\section{Test-Optional Admissions}
\label{sec:analysis}

\subsection{Optimal Acceptance Rule}

In a test-optional regime, our college has two instruments: the imputation rule and the acceptance rule. Only the imputation  rule affects students' score submission, and in turn the college's and society's information. Moreover, the only decision that students make is whether to submit their score. So, no matter the imputation rule, the college's optimal acceptance rule simply maximizes its ex-post utility given students' submission behavior. Formally, recalling that $L(\tau(x)|x)$ is the average test score of non-submitters with observable $x$ given the imputation $\tau(x)$:

\begin{lemma}
\label{lem:prelim}
Consider test optional with any imputation rule $\tau$. The college has an optimal acceptance rule in which a student with observable $x$ and imputed/submitted score $\hat t$ is accepted if (i) $\hat t>\tau(x)$ and $u^*(x,\hat t)>0$ or if (ii) $\hat t=\tau(x)$ and $u^*(x,L(\tau(x)|x))>0$, and is rejected otherwise. 
\end{lemma}

Any optimal acceptance rule must have the college making ex-post optimal decisions on path. The lemma also specifies rejecting any student who has a test score below the imputed level but who chooses, off path, to submit. When the non-submitters are accepted, 
we could replace this behavior with any other monotonic rule and the outcome would be the same. When the non-submitters are rejected, though, monotonicity of the admission rule requires the college to also reject any score submission below the imputed score. In this latter case, commitment to the policy may be necessary: off path, the college may be rejecting students that it ex-post prefers to accept.   For example, suppose test scores at some observable $x$ are distributed uniformly between 0 and 100, and the imputation is $\tau(x) = 50$. Students with scores between 0 and 50 don't submit, leading to an average score of 25 for non-submitters.  If the college's ex-post bar for acceptance is in between 25 and 50, say $\u t^*(x) = 40$, then the college will reject the non-submitters. The college must then reject all off-path submissions of scores below 50, including---ex-post suboptimally---those above its ex-post bar of 40.

\autoref{lem:prelim} allows one to deduce, at any given observable $x$, the college's value from inducing any belief (the expected test score), say $\tilde t$. This is because the lemma tells us what acceptance decision $A$ the college would make, and we can plug that into \eqref{eq:U^c} with $t^c=t^s=\tilde t$. But even so, the college's problem cannot be solved with standard information-design tools because of a constraint on what information the college can generate: it can only choose an imputation and then rely on the students' score disclosure.

\subsection{Flexible Imputation}
\label{sec:flexible}

We now turn to studying optimal admission policies under flexible imputation. Clearly, the college can ensure that it is no worse off than under test mandatory: after all,  the imputation rule $\tau(\cdot)=-\infty$ ensures that all students submit their scores. But when and how can the college do better?

In choosing its imputation $\tau(x)$ for some observable $x$, the college trades off making better admission decisions with reducing disagreement cost. Raising $\tau(x)$ leads fewer students to submit their test scores. The cost is that the college now has less information with which to make admissions decisions. The benefit is that by pooling together a larger set of test scores (those of the non-submitters), the college can reduce the disagreement cost it bears with society, as we saw in \autoref{sec:motivatingex}. In particular, consider two students who are both rejected or both accepted. If their test scores are either both below society's bar $\u t^s(x)$ or both above, the disagreement cost is the same regardless of whether these students submit their scores or are pooled together. But if these students are on opposite sides of society's bar, then the disagreement cost is lower when the students are pooled together.

When solving for the optimal admissions policy, the college's problem is separable across observables. That is, we can optimize at each observable $x$ and then ``stitch'' together the solutions across $x$'s to get the globally optimal admission policy. 

Given some fixed $x$, it is useful to consider separately the case in which the college is less selective than society ($\u t\Col(x) <\u t^*(x) < \u t\Soc(x)$) and the case in which it is more selective ($\u t\Soc(x) < \u t^*(x) < \u t\Col(x)$).\footnote{The remaining case, $\u t^c(x)=\u t^*(x)=\u t^s(x)$, is trivial, as there is no disagreement at the observable $x$. The first-best is achieved when the college uses imputation $\tau(x)=\u t^*(x)$ and accepts a student if and only if they submit a score $t>\tau(x)$.} For both cases, we will assume without loss that the imputation level $\tau(x)$ is set as $\tau(x) \geq \u t^*(x)$, and that any submitted score $t > \tau(x)$ is accepted.\footnote{Suppose the college were to reject imputed/submitted scores up to some threshold $t'>\tau(x)$. Then it could instead raise the imputation level to $t'$, still reject non-submitters, and now accept all submitted scores. This alternative policy leads to the same admission decisions but weakly lowers disagreement costs by pooling a superset of scores. Given that the college accepts any submitted score $t>\tau(x)$, \autoref{lem:prelim} implies that $\tau(x)\geq \u t^*(x)$.}

\paragraph{College is less selective than society.} 

When the college is less selective, setting $\tau(x) = \u t^*(x)$ and rejecting non-submitters replicates not only the test-mandatory admission decisions, but also the college's test-mandatory payoff. This is because all of the scores being pooled together are below society's acceptance threshold $\u t^s(x)$. Furthermore, the college does worse if it sets $\tau(x)>\u t^*(x)$ and then rejects non-submitters: it is now rejecting students that it preferred to accept even if it had to pay a disagreement cost to do so. Altogether, if the college rejects non-submitters, then it cannot improve on setting $\tau(x)= \u t^*(x)$ and replicating the test-mandatory outcome.

The college might improve on test mandatory, however, by {accepting} non-submitters at some observable. Monotonicity of the acceptance rule means that the college would then accept all students with this observable. With all of these students being accepted, the college would minimize disagreement costs by setting the imputation level to infinity, so that none of these students submit scores.\footnote{If $\E[t|x]>\u t\Soc(x)$, then any large enough $\tau(x)$ would also be optimal as that would ensure that society prefers to accept the pool of non-submitters, resulting in zero disagreement cost.} Of course, relative to test mandatory, the college would then be admitting too many low-scoring students. Hence:

\begin{proposition}
\label{prop:lessselective}
Consider flexible imputation and some observable $x$. When the college is less selective than society ($\u t\Col(x) < \u t^*(x) < \u t\Soc(x)$), it is optimal for the college to either:
\begin{enumerate}
\item \label{prop:lessselective-admitall} Impute $\tau(x)=\infty$ and accept students regardless of imputed/submitted score $\hat t$; or 
\item \label{prop:lessselective-replicate} Replicate the test-mandatory outcome by imputing $\tau(x) = \u t^*(x)$, rejecting students with imputed/submitted score $\hat t \le \u t^*(x)$,  and accepting students with $\hat t > \u t^*(x)$.
\end{enumerate}
\end{proposition}

\autoref{fig:more_eager_accept} and \autoref{fig:more_eager_reject} illustrate the two possibilities. 

\begin{figure}[p]

\caption{College's test-optional payoff as a function of the imputed test score, 
when the college is less selective than society. \label{fig:EUontauLessSelective}}

\begin{centering}

\begin{subfigure}{1 \textwidth}

\caption{College's payoff is maximized by setting $\tau =\infty$ (or any $\tau \geq 80$)
and accepting non-submitters. 
\label{fig:more_eager_accept}}

\[\includegraphics[width = 4.55 in]{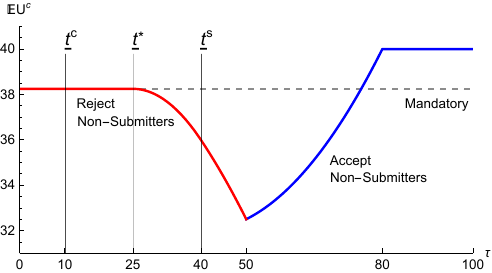}\]

\end{subfigure}

\bigskip \bigskip

\begin{subfigure}{1 \textwidth}

\caption{College's payoff is maximized by setting $\tau\leq 55=\u t^*$ 
and rejecting non-submitters.
\label{fig:more_eager_reject}}

\[\includegraphics[width = 4.55 in]{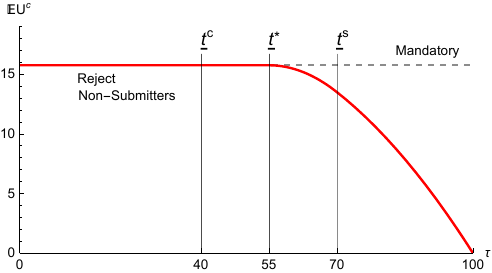}\]

\end{subfigure}

\bigskip

\parbox{.95 \textwidth}{
\footnotesize Fix  some observable $x$. The distribution of test scores given $x$ is $t\sim U[0,100]$. Utilities are $u\Col(x,t) = t - \underline t \Col$, $u\Soc(x,t) = t - \underline t \Soc$, and $\delta = 1$, implying $\underline t^* = (\underline t \Col + \underline t \Soc)/2$. Along the red portion of the curves, the college prefers to reject non-submitters; along the blue portion, the college prefers to accept them.
}

\end{centering}

\end{figure}

\paragraph{College is more selective than society.} Let us turn to observables at which the college is more selective than society. Unlike when the college is less selective, the college can improve on test mandatory by imputing the ex-post optimal bar, rejecting non-submitters, and accepting submitters. Pooling together the scores of all the rejected students now reduces disagreement cost because society prefers to reject some of those students (those with  $t< \u t^s(x)$) and accept others ($t \in (\u t^s(x),\u t^*(x)$). In general, the college might do even better by choosing a higher imputation, altering the set of admitted students.

\begin{proposition}
\label{prop:moreselective}
Consider flexible imputation and some observable $x$. When the college is more selective than society ($\u t\Soc(x) < \u t^*(x) < \u t\Col(x)$), the college optimally chooses imputation $\tau(x)\in [\u t^*(x), \u t\Col(x)]$; it rejects students with imputed/submitted score $\hat t \le \tau(x)$ and it accepts students with $\hat t>\tau(x)$. 
\end{proposition}

The proposition's proof establishes that the optimal $\tau(x)$ is determined by comparing the function $L(\cdot|x)$, which gives the average test score of non-submitters, with society's bar $\u t\Soc(x)$. Specifically, letting $t^\circ$ be a score at which $L(t^\circ|x)=\u t\Soc(x)$,\footnote{If $L(\cdot|x)$ is everywhere below $\u t\Soc(x)$ then let $t^\circ = \infty$, and if $L(\cdot|x)$ is everywhere above $\u t\Soc(x)$ then let $t^\circ = -\infty$. Otherwise, for simplicity of discussion, we assume there is a solution to 
$L(t^\circ|x)=\u t\Soc(x)$, as is guaranteed when the distribution of $t|x$ is atomless.} the college sets
$$
\tau(x)=
\begin{cases}
\u t^*(x) & \text{ if } t^\circ \leq \u t^*(x)\\
t^\circ & \text{ if } t^\circ \in (\u t^*(x),\u t\Col(x))\\
\u t\Col(x) & \text{ if } t^\circ \geq \u t\Col(x)
\end{cases}.
$$

For the intuition behind \autoref{prop:moreselective}, consider the case in which $t^\circ \in (\u t^*(x), \u t^c(x))$. The optimal admissions policy then involves setting $\tau(x) = t^\circ$, rejecting non-submitters, and accepting submitters.\footnote{To see why this acceptance policy is optimal given the imputation $\tau(x) = t^\circ$, notice that disagreement cost is zero regardless of whether non-submitters are accepted or rejected, because $L(\tau(x)|x) = \u t\Soc(x)$. Since $t^\circ<\u t\Col(x)$, it is better for the college to reject non-submitters at this imputation level. It is better to accept submitters, on the other hand, because $t^\circ>\u t^*(x)$.}  This imputation makes society indifferent over whether to accept the pool of non-submitters, as their expected test score is $L(\tau(x)|x) = \u t\Soc(x)$. Moreover, society wants to accept any submitter, since their score is $t\geq \tau(x)>\u t\Soc(x)$. So the disagreement cost is zero.
Now consider a marginal change of the imputation level $\tau(x)$ from $t^\circ$ to $ t'$.
On the one hand, raising the imputation level $\tau(x)$ to $ t'> t^\circ$ cannot help. Doing so and then rejecting the larger pool\footnote{For any marginal change, the college will still prefer to reject the pool, since the expected test score of non-submitters is strictly below $\u t^*(x)$.} yields the same set of admitted students and the same disagreement cost as setting $\tau(x) = t^\circ$ and then rejecting students with scores $t \in (t^\circ,  t']$; there is no benefit from pooling the scores of these marginal students with those below $t^\circ$ since society does not strictly prefer to reject the pool of non-submitters. But the latter policy is dominated by the originally proposed policy of setting $\tau(x) = t^\circ$ and accepting students with scores $t \in (t^\circ,  t']$, as they provide positive ex-post utility.  On the other hand, lowering the imputation level to $ t'< t^\circ$ also cannot help. Doing so and then rejecting students with $t \in (t', t^\circ]$ yields the same set of admitted students but higher disagreement cost, since society strictly prefers to reject the pool of non-submitters when $\tau(x)= t'$;  doing so and then accepting students with $t \in (t', t^\circ]$ yields a worse set of admitted students from the college's perspective, as $t<\u t\Col(x)$, but identical (zero) disagreement cost.

\autoref{fig:EUontauMoreSelective} illustrates two examples of \autoref{prop:moreselective}. Panel \ref{fig:less_eager_interior} shows a case in which the optimal $\tau(x)$ is in $(\u t^*(x), \u t^c(x))$. Panel \ref{fig:lesseagerfirstbest} shows a case in which the optimal $\tau(x)$ is equal to $\u t^c(x)$, and the college achieves its first best: it accepts students if and only if $t >\u t^c(x)$, and it incurs no disagreement cost. 
Although not illustrated in the figure, it is also possible that the optimal $\tau(x)=\u t^*(x)$.

\begin{figure}[p]

\caption{College's test-optional payoff as a function of the imputed test score, 
when the college is more selective than society.\label{fig:EUontauMoreSelective}}

\begin{centering}

\begin{subfigure}{1 \textwidth}

\caption{College's payoff is maximized by setting $\tau=50\in (\u t^*, \u t^c)$
and rejecting non-submitters. 
\label{fig:less_eager_interior}}

\[\includegraphics[width = 4.55 in]{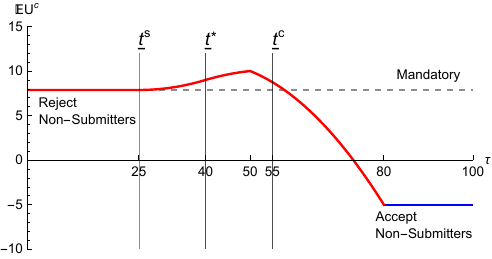}\]

\end{subfigure}

\bigskip \bigskip

\begin{subfigure}{1 \textwidth}

\caption{College achieves its first best by setting $\tau=70=\u t^c$ 
and rejecting non-submitters.
\label{fig:lesseagerfirstbest}}

\[\includegraphics[width = 4.55 in]{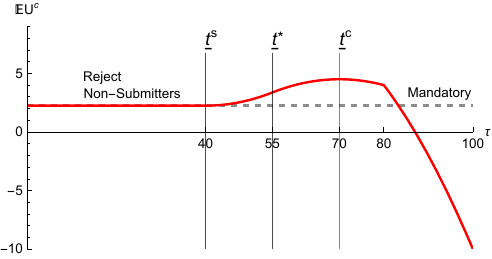}\]

\end{subfigure}

\bigskip

\parbox{.95 \textwidth}{
\footnotesize Fix  some observable $x$. The distribution of test scores given $x$ is $t\sim U[0,100]$. Utilities are $u\Col(x,t) = t - \underline t \Col$, $u\Soc(x,t) = t - \underline t \Soc$, and $\delta = 1$, implying $\underline t^* = (\underline t \Col + \underline t \Soc)/2$. Along the red portion of the curves, the college prefers to reject non-submitters; along the blue portion, the college prefers to accept them.
}

\end{centering}

\end{figure}

\paragraph{How are students affected?} The outcomes of a college-optimal admissions policy under test-optional admissions have clear-cut and intuitive implications for student welfare relative to the outcomes of test-mandatory admissions.

 Students benefit from test optional at observables where the college is less selective than society. Specifically, at these observables, \autoref{prop:lessselective} implies that either the college replicates the test-mandatory admissions, or it admits all students. In the latter case, high-scoring students (with $t>\u t^*(x)$)) are indifferent between test optional and test mandatory, but low-scoring students ($t < \u t^*(x)$) strictly benefit.

By contrast, students are harmed by test optional at observables where the college is more selective than society. Specifically, \autoref{prop:moreselective} implies that when the optimal imputation is $\tau(x)=\u t^*(x)$, the test-mandatory outcome is replicated for all students. But when the optimal imputation is $\tau(x)>\u t^*(x)$, intermediate-scoring students (with $t \in (\u t^*(x),\tau(x)]$) are rejected under test optional while they would have been accepted under test mandatory, whereas the outcomes for low- and high-scoring students ($t<\u t^*(x)$ and $t>\tau(x)$, respectively) are unchanged.

\paragraph{How is society affected?} Relative to test mandatory, a move to test optional with flexible imputation harms society (at least weakly) at every observable $x$. At observables where the college is less selective than society, additional students are now admitted, while the college was already admitting all students society wants admitted. At observables where the college is more selective, fewer students are now admitted, while the college was already rejecting all students society wants rejected. 

In addition, under test optional with flexible imputation, society can be worse off when the social-pressure intensity $\delta$ increases (unlike with test mandatory). 
Consider the example from \autoref{fig:more_eager_accept}, where $\delta=1$ and at the relevant observable the college optimally sets $\tau=\infty$ (effectively test blind), accepting all students. If the intensity drops to $\delta\approx 0$, then it becomes optimal for the college to set $\tau=-\infty$ (effectively test mandatory) and reject students with scores below the ex-post bar. Since society is more selective than the college, society is better off when $\delta\approx 0$ than $\delta=1$. So, concretely, if social pressure reflects the current student body's preference to admit fewer legacy applicants, then increasing pressure could be counterproductive.

\subsection{Restricted Imputation}
\label{sec:restricted}

We now turn to test-optional admissions when the imputation rule $\tau(\cdot)$  is exogenously given. The college only optimizes its acceptance rule. As discussed in \autoref{sec:model-restrictedimputation}, many colleges announce a policy that we interpret as no adverse inference imputation. Restricted imputation also subsumes test-blind admissions, as that is equivalent to $\tau(\cdot)=\infty$.

\paragraph{The optimal acceptance rule.} 

As with flexible imputation, we can solve for an optimal acceptance rule under restricted imputation separately for each observable $x$. An optimal acceptance rule readily follows from \autoref{lem:prelim}:

\begin{proposition} \label{prop:restrictedopt}
Consider some observable $x$ and imputation level $\tau(x)$. An optimal acceptance rule for the college is as follows. A student with submitted score $t>\tau(x)$ is accepted if and only if $t>\u t^*(x)$; a student with submitted score $t<\tau(x)$ is rejected; and a student with imputed/submitted score $\tau(x)$ is accepted if and only if $L(\tau(x) | x)>\u t^*(x)$.
 \end{proposition}

The proposition says that the college's acceptance rule on path is determined by comparing a student's expected score---the score if submitted, or $L(\tau(x)|x)$ if not submitted---with the ex-post bar. (Submission of $t\leq \tau(x)$ only occurs off path.) Whether the college is more or less selective than society does not affect the college's optimal acceptance rule; the distinction matters under flexible imputation (\autoref{sec:flexible}) only because it affects the optimal imputation.

To better understand the admissions policy under restricted imputation, we can consider exogenously varying the imputation $\tau(x)$ at a given $x$. In that case, there is a threshold $T(x)$ such that if $\tau(x) < T(x)$, then it is optimal to reject non-submitters, whereas if $\tau(x)>T(x)$, then it is optimal to accept non-submitters.\footnote{$T(x)\geq \min\{\underline t{\Soc}(x),\underline t{\Col}(x)\}$, implying that if the imputation is below both the college's and society's bars, then it is optimal to reject non-submitters. In fact, $T(x)=\infty$ if $\E[t|x]\leq \u t^*(x)$. If $\E[t|x] > \u t^*(x)$, then so long as the distribution of test scores conditional on $x$ has full support and is atomless, $T(x)$ is the unique solution to $L(T(x)|x)=\u t^*(x)$.} \autoref{fig:EUontauLessSelective} and \autoref{fig:EUontauMoreSelective} illustrate, at some fixed observable $x$, how the college's payoff and its decision of whether to accept non-submitters may depend on the imputation level.

\paragraph{How are students affected?}  Whether students at an observable $x$ benefit from test optional under restricted imputation (relative to test mandatory) depends on how the imputation level $\tau(x)$ and the lower expectation $L(\tau(x)|x)$ compare with the ex-post bar $\u t^*(x)$. To understand how these vary with observables, we must make further assumptions.

\label{increasing_observables}Accordingly, say that there is an ordered \emph{subset of increasing observables} $\mathcal X'$, which we (re-)label to be on the real line ($\mathcal X'\subset \Reals$), if on $x \in \mathcal X'$ it holds that: (i) $u\Col(x,t)=v\Col(x)+t$ and $u\Soc(x,t) = v\Soc(x)+t$, with  $v\Col(x)$ and $v\Soc(x)$ both increasing; (ii) the distribution of test scores has the monotone likelihood ratio property (MLRP),\footnote{That is, there is a test-score density/probability mass function $f(t|x)$ such that for each $x'>x''$ in  $\mathcal X'$ and each $t'>t''$, it holds that $f(t'|x')f(t''|x'')\geq f(t'|x'')f(t''|x')$.} and (iii) $\tau(x)$ is increasing. 
Property (i) guarantees that the ex-post bar $\u t^*(x)$ is decreasing in the observable, while properties (ii) and (iii) guarantee that the expected test score conditional on not submitting, $L(\tau(x) | x)$, is increasing. Property (iii) is implied by property (ii) when $\tau$ is the no adverse inference rule defined by $\tau(x) = \mathbb{E}[t|x]$. 

A subset of increasing observables yields straightforward implications for which students benefit or are harmed by test-optional admissions with restricted imputation, as can be seen using \autoref{fig:increasingobs}. As detailed in the figure's subcaption, students with ``low'' observables (those with $x$ such that $\tau(x)<\u t^*(x)$) are unaffected. Those with ``medium'' observables ($x$ such that $\tau(x)>\u t^*(x)$ but $L(\tau(x)|x) < \u t^*(x)$) are harmed: within this group, high- and low-scorers ($t > \tau(x)$ and $t\le \u t^*(x)$, respectively) are unaffected while medium-scorers lose  (they are only admitted under test mandatory). Students with ``high'' observables ($x$ such that $L(\tau(x)|x) > \u t^*(x)$) benefit: in this group, low-scorers win (they are only accepted under test optional) while the rest are unaffected.

\begin{figure}
\caption{How test-optional with restricted imputation affects students on a subset of increasing observables. 
\label{fig:increasingobs}}

	\[\includegraphics[width = 3.5 in]{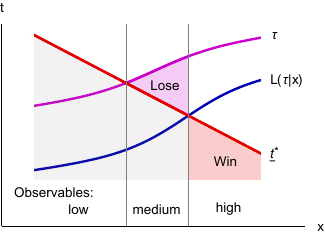}\]

\parbox{.95 \textwidth}{
\footnotesize Under test optional, a student submits their score if $t >\tau$, and a non-submitter is admitted if $L>\underline t^*$. Under either test optional or mandatory, a submitter is admitted if $t>\underline t^*$. So students in the lightly shaded region are rejected under both regimes, while students in the unshaded region are accepted under both. Students in the ``Win'' region are only accepted under test optional; students in the ``Lose'' are only accepted under test mandatory.
}

\end{figure}

\paragraph{Restricted vs.~flexible imputation.} Under restricted imputation, given a subset of increasing observables, students with good observables benefit under test optional while students with medium observables are harmed. By contrast, under flexible imputation, it is students with observables at which the college is less selective than society that benefit and those with observables at which the college is more selective that are harmed. Looking back at \autoref{fig:upics}, we see that these predictions may go in the same qualitative direction, or may go in opposite directions.\footnote{For the example in \autoref{fig:upics}, utilities were defined as $u^i(x,t) = a^i+x + w^i t$, with $x \in \mathbb{R}$ and $w^i>0$; we can rescale these utilities as $u^i(x,t) = \frac{a^i}{w^i}+\frac{x}{w^i} + t$.
So observables are increasing when test scores satisfy the MLRP and $\tau(x)$ is increasing.} In the figure's left panel, where the college weights tests less than society, the college is less selective than society at higher observables. Hence, students with higher observables benefit from test optional under both flexible and restricted imputation. In \autoref{fig:upics}'s right panel, where the college weights tests more than society, we have the reverse: the college is more selective at higher observables. In this case, the predictions about which students benefit from test optional flip depending on whether imputation is flexible or restricted.

Recall that under flexible imputation, a switch from test mandatory to test optional (weakly) benefits the college and harms society.
These effects hold at every observable $x$. By contrast, under restricted imputation, either party may benefit or be harmed by the switch
at any specific $x$---it depends on the imputation level $\tau(x)$.\footnote{The effects need not go in opposite directions for the two parties. Consider the example in \autoref{fig:lesseagerfirstbest}. Society is harmed by any imputation $\tau(x)>\underline{t}^*(x)$, as that leads to fewer students being admitted under test optional. When $\tau(x)$ is sufficiently large, the college is also harmed.} The effects are thus also ambiguous when aggregating across observables. Nevertheless, we can make two general observations. First, under test optional with a fixed imputation rule, society is better off at every observable when the social-pressure intensity $\delta$ increases. This is not true under flexible imputation because it can trigger changes in the imputation rule and alter which students submit scores. But when the set of submitting students is held fixed, an increase in $\delta$ only pulls the college's admission decisions towards what society prefers.

Second, there is a sense in which the college strictly prefers test mandatory when either social pressure is limited or the college's and society's preferences are similar. To make that precise, say that \emph{information is valuable} for the college at a given imputation rule if in the absence of any disagreement cost ($\delta=0$), the college would strictly prefer test mandatory to test optional with that imputation rule.\footnote{This is a mild requirement because the college always weakly prefers test mandatory absent disagreement cost. A strict preference only requires that there is a positive-probability set of observables for which test optional with the restricted imputation does not replicate the admissions outcome of test mandatory.}

\begin{proposition}
\label{prop:alignedprefs}
    Fix any imputation rule. The college strictly 
    prefers test mandatory to test optional if either:
\begin{enumerate}
\item \label{aligned1} Information is valuable for the college, and social-pressure intensity $\delta$ is sufficiently close to zero; or
 \item \label{aligned2} Information would be valuable for the college if it shared society's preferences, and the college's preferences are sufficiently close to society's.\footnote{Sufficiently close in the metric $\sup_{x \in\mathcal{X}} |\u t\Col(x) - \u t\Soc(x)|$.}
\end{enumerate}
\end{proposition}

The result is mathematically straightforward, as it is no more than a statement of continuity. But its forces may speak to why a small number of colleges \href{https://www.usnews.com/education/best-colleges/applying/articles/some-colleges-are-requiring-test-scores-again-what-it-means-for-applicants}{have recently reverted} to test-mandatory admissions. Part \ref{aligned1} of \autoref{prop:alignedprefs} offers one explanation in terms of diminished social-pressure intensity. More interestingly,  some of the reverting institutions---e.g., \href{https://president.dartmouth.edu/news/2024/02/reactivating-satact-requirement-dartmouth-undergraduate-admissions}{Dartmouth} and \href{https://hub.jhu.edu/2024/08/16/undergrad-admissions-standardized-test-requirement/}{John Hopkins}---have referenced how their own data from their test-optional period has informed them about the value of test scores; see \citet{FCT24} for a study using multiple Ivy-Plus colleges. We can interpret these colleges' preferences on how to weigh test scores as becoming more similar to society's (recall the \href{https://www.pewresearch.org/fact-tank/2022/04/26/u-s-public-continues-to-view-grades-test-scores-as-top-factors-in-college-admissions/}{2022 PEW research survey} mentioned in our introduction). The reversal to test mandatory is then consistent with \autoref{prop:alignedprefs} part \ref{aligned2}, so long as their imputation rule is fixed---e.g., to no adverse inference or some variant thereof. 

\paragraph{Constrained imputation.} In practice, colleges may not be fully restricted nor fully flexible in their imputation. They may optimize subject to constraints reflecting another dimension of social pressure or internal politics. One natural constraint would be that a college has to use a ``monotonic'' imputation rule, say imputing a higher test score to an otherwise-identical student who has a higher GPA. An additional constraint could be that the imputation rule has to be ``continuous'' in observables.

Formally, consider our preceding formulation of \hyperref[increasing_observables]{increasing observables}. Assume the primitives satisfy those properties (i) and (ii), and the college can choose from some set of imputation rules that satisfies property (iii).  In words, both the college and society prefer higher observables, students with higher observables have better test scores, and the college must impute higher scores for non-submitters with higher observables. As long as the college can choose the constant imputation $\tau(\cdot)=-\infty$ that replicates test mandatory, the college benefits from being test optional, as under flexible imputation. But in terms of student welfare, our conclusions under restricted imputation apply.

\section{Further Results}
\label{sec:extensions}

\subsection{Banning Affirmative Action}

\label{sec:AA}

In this section, motivated by public debates about affirmative action, we apply our framework to show how banning affirmative action can push a college from test-mandatory admissions to test-blind admissions.\footnote{We study test blind rather than test optional for simplicity; as noted previously, test blind is equivalent to test optional when non-submitters are imputed sufficiently high test scores.} Conceptually, this serves as an extended example in which we can say whether a college benefits from test optional for a given restricted imputation. The formal statements and proofs of the results for this section are in \appendixref{app:AA}.

The setup is as follows. Students come from one of two groups, with group identity labeled as $x_0 \in \{r, g\}$---red or green. Each student also has an observable $x_1 \in \mathbb{R}$, which can represent some aggregate of GPA and/or extracurriculars, and a test score $t \in \{0,1\}$. The college and society agree on the importance of $t$ and $x_1$, but the college  has a preference for admitting group $g$ and society does not. Specifically, 
\begin{align*}
u\Soc(x,t) &= x_1 +  t, \\
u\Col(x,t) &= x_1 +  t + \beta  \ind_{x_0 = g}-c,
\end{align*}
with  $\beta>c>0$, and $\ind_{x_0=g}$ an indicator for $x_0=g$. So the college has a lower bar than society for green types and a higher bar for red types.
Our leading interpretation is that group identity $x_0$ corresponds to race, with the college putting an explicit weight on enrolling underrepresented minorities ($x_0=g$) while society prefers race-blind admissions.

We assume that the observable $x_1$ has the same distribution for both groups; for tractability, we stipulate a uniform distribution over a large enough interval. The test score, by contrast, can be correlated with the group identity $x_0$. We assume the distribution of test scores depends on student characteristics only through $x_0$: $\Pr(t=1 | x=(x_0,x_1))=p_{x_0} \in (0,1)$. Following the interpretation of group identity as race, assume the underrepresented group $g$ has a lower average test score: $p_g<p_r$.

We consider two affirmative-action regimes. If affirmative action is allowed, the college may condition admissions on group identity $x_0$. If affirmative action is banned, group identity is hidden and therefore cannot be used in admissions.

We find that when affirmative action is allowed, the college prefers test mandatory (\autoref{prop:aaallowed}). Given the test score and other observables, the college sets a lower admission threshold on $x_1+t$ for its preferred group.  When affirmative action is banned, identity in its admissions rule. But 
differential thresholds by group are no longer feasible. Test scores now become a signal of group identity (as in  \citet{chan2003does} and others). Because the college's preferred group $g$ has a lower average test score, the college now effectively wants to put a lower weight on tests than does society. Thus, banning affirmative action can push a college to switch to test blind (\autoref{prop:aabanned}). The switch is more likely when social-pressure intensity $\delta$ is larger, the gap in test scores $p_r-p_g$ is larger, or the college puts a larger weight $\beta$ on group identity (\autoref{cor:aacompstats}). 

We also find that holding fixed the college's testing regime (mandatory or blind), society prefers to ban affirmative action. After all, society does not want group identity factored in admissions. Moroever, given any affirmative-action regime, society prefers the college to be test mandatory rather than test blind. But since banning  affirmative action can trigger the college to ignore test scores, a ban can backfire and make society worse off (\autoref{prop:society}).

As of June 2023, affirmative action by race has been banned nationwide by the US Supreme Court.
Even prior to this decision, some commentators had suggested that banning affirmative action might induce colleges to avoid seeing test scores 
(e.g., \href{https://www.newyorker.com/news/daily-comment/the-supreme-court-appears-ready-finally-to-defeat-affirmative-action}{New Yorker, January 2022}). One rationale is that going test blind can suppress evidence of score differentials across groups, which could have been used in lawsuits alleging that a college makes illegal decisions based on group identity. Our story is distinct, but complementary. We assume that when a college cannot directly condition on race, it wants to put less weight on tests that are correlated with race. Hiding test scores allows the college to do so in a way that generates less disagreement with society. This can be interpreted as protecting the college from criticism of how much weight it places on different elements of an application.

We conclude this subsection by noting that group identity in our extended example can also  represent other attributes besides race. For example, it could be legacy status. Colleges often give a leg up in admission to students whose family members have attended the college. The above analysis then applies for evaluating a ban on the use of legacy status in admissions, with one possible twist. For legacy status, the preferred group might have a higher average test score than other applicants, instead of a lower one, as legacy students tend to come from advantaged backgrounds.\footnote{For example, the \href{https://features.thecrimson.com/2023/freshman-survey/academics/}{Harvard Crimson reports} that legacy students in Harvard's entering class of 2027 had an average SAT score of about 28 points higher than their non-legacy counterparts.} A ban on legacy admissions---which has been implemented by a few US states, most recently California in September 2024---would then mean that the college now wants to put \emph{extra} weight on tests. Hence, the college's desire to see test scores would increase.

\subsection{Competition Between Colleges}
\label{sec:competition}

So far, we have considered a single college. We now extend our model to multiple colleges, making two points. First, when a test-optional college competes against a test-mandatory counterpart, it may encounter adverse selection. Among its nonsubmitting applicants, the ones with lower test scores are less likely to be admitted by the competitor. So if the test-optional college were to admit nonsubmitters, the worse applicants in the pool would be more likely to matriculate. Second, the consequences of such adverse selection are ambiguous. The test-optional college now has a lower benefit of accepting nonsubmitters. But, more subtly, the social-pressure cost of rejecting nonsubmitters can also decrease. A college may thus be more or less inclined to go test optional when its competitor is test mandatory.

To develop these points formally, we introduce the notion of \emph{yield}.
From the perspective of any one college, a student's yield $y\in [0,1]$ is the probability that the student attends that college conditional on being admitted. 
This yield will depend on whether a student is admitted to other colleges, and thus on the admission policies of a college's competitors. To incorporate yield into payoffs, we stipulate, naturally, that admitting a student with observables $x$, score $t$, and yield $y$ provides the college underlying (expected) utility $y\times u^c(x,t)$. 
Our key assumption is that social pressure also scales analogously.\footnote{Formally, we modify \autoref{eq:U^c} so that the college's payoff for making admission decision $A$ for a student of type $(x,t)$ with yield $y(x,t)$ is $\left(A u^c(x,t)-\delta d(x,t^s,A)\right) y(x,t)$. Assuming for convenience that for each observable $x$, the distribution of test scores has a density $f(t|x)$, society now evaluates a nonsubmitting student by $t^s=L^Y(\tau(x)|x)$, where $$
L^Y(\tau(x)|x):= \frac{\int_{t\le \tau(x)} t y(x,t) f(t|x) \mathrm dt}{\int_{t\le \tau(x)} y(x,t) f(t|x) \mathrm dt}$$ is the yield-weighted average test score of nonsubmitters. Note that if $y(x,t)$ is constant on $t\leq \tau(x)$, then $L^Y(\tau(x)|x)$ simplifies to $L(\tau(x)|x)\equiv \E[t|t\leq \tau(x),x]$, which is what we had without yield.}
For applicants who would certainly attend a college  if admitted ($y=1$), we recover our original social pressure costs. At the opposite limit, society doesn't judge a college's decision at all for a student  who never would have come anyway ($y=0$).

Since underlying benefits and social pressure costs both scale proportionally with yield, a test-mandatory college will admit a student with arbitrary yield $y\in[0,1]$ if and only if it would admit that student at yield $y=1$. Consequently, a test-mandatory college's admission decisions are independent of its competitors' policies. 

By contrast, yield can affect a college's admissions when it is test optional. The college must then consider the possibility of differential yield within an observably identical pool of nonsubmitters. The college might be willing to admit a pool based on the {average} student, but not the \emph{yield-weighted average} student, if students with higher unsubmitted test scores have lower yield. Such adverse selection can happen when the college has test-mandatory competitors who only admit students with higher test scores.

To illustrate the adverse selection concretely, consider two identical colleges. These colleges have common underlying utility $u^c$, are judged by the same societal utility $u^s$, and face the same social pressure intensity $\delta$. Both colleges are required to use the same restricted imputation rule $\tau$ if test optional. The colleges are also equally attractive to students: any student admitted at only one college will attend that college, while a student admitted at both will choose uniformly at random between the two.

\begin{figure}[h]
\caption{A test-optional college's yield varies with its competitor's testing regime. 
\label{fig:compet}}

\begin{centering}

\begin{tabular}{cc}
\begin{subfigure}{0.45\textwidth}
\includegraphics[width = 2.75 in]{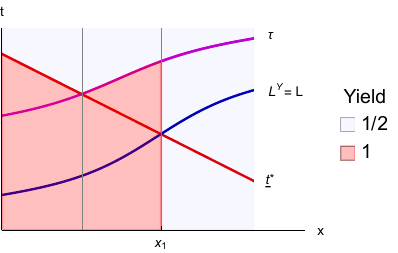}
\caption{Test-optional competitor \label{subfig:competTO}}
\end{subfigure}
&
\begin{subfigure}{0.45\textwidth}
\includegraphics[width = 2.75 in]{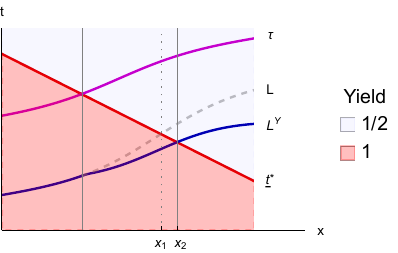}
\caption{Test-mandatory competitor \label{subfig:competTM}}
\end{subfigure}
\end{tabular}

\end{centering}
\label{fig:prefs}

\medskip
\parbox{.95 \textwidth}{
\footnotesize The yield-weighted average test score of nonsubmitters is $L^Y$. For comparison, the unweighted average $L$ is also shown. Students withhold test scores when $t\le \tau$. At  observable $x$, the college admits nonsubmitters if $L^Y$ is above the ex-post bar $\underline t^*$. So nonsubmitters are admitted for $x>x_1$ in the left panel and $x>x_2$ in right panel.
}
\end{figure}

\autoref{fig:compet} illustrates how a test-optional college's admission decisions can vary with its competitor's admission policy. 
The figure assumes increasing observables in the sense of \autoref{sec:restricted} and \autoref{fig:increasingobs}.
The left panel, \autoref{subfig:competTO}, considers the college's problem when the competitor is test optional.\footnote{We assume the competitor uses its optimal admission policy based on the unweighted yield ($L$ in the figure), which is equivalent to the optimal policy absent competition. This policy is indeed a best response.}
At any given observable $x$, all of the nonsubmitters  have the same yield because the competitor admits them all or rejects them all. Since there is no differential yield, there is no adverse selection: the yield-weighted average test score of nonsubmitters $L^Y$ is the same as the unweighted average test score $L$. So the college makes admission decisions as in our single-college analysis, unaffected by yield. It accepts nonsubmitters with observables $x>x_1$.

Now consider the right panel, \autoref{subfig:competTM}, which depicts a test-mandatory competitor. Here there are observables at which the competitor selectively admits nonsubmitters  with relatively high test scores---only those above the ex-post bar $\underline t^*$. So the college faces a low yield of $1/2$ for high-scoring students and a high yield of $1$ for low-scoring students. This adverse selection lowers the yield-weighted average test score, bringing $L^Y$ below $L$. The college is now less inclined to admit nonsubmitters than when it faced a test-optional competitor; the threshold $x$ at which it accepts nonsubmitters increases from {$x_1$ to $x_2$}.

\autoref{subfig:competTM} also helps us understand why adverse selection isn't necessarily bad for a test-optional college. For nonsubmitters who are accepted both before and after a pool gets worse ($x>x_2$), making the pool worse clearly hurts the college: it decreases underlying payoffs and cannot decrease social pressure. However, consider nonsubmitters who are rejected both before and after ($x<x_1$). Now, making the pool worse has no impact on underlying payoffs but it can decrease social pressure if society had wanted to admit the pool. So at observables where nonsubmitters are rejected and a less-selective society wants the nonsubmitters accepted, a test-optional college can benefit from its competitor being test-mandatory.

The foregoing discussion suggests that, under restricted imputation, a college's preferred testing regime can depend on its competitor's choice. One possibility is strategic complementarity: a college prefers test mandatory when its competitor is test mandatory, and prefers test optional when its competitor is test optional. Another is strategic substitutability: a college prefers test optional when its competitor is test mandatory, and prefers test mandatory when its competitor is test optional. 

In \appendixref{app:competition}, we provide examples to illustrate how both of those possibilities can arise from adverse selection. In our complementarity example, a test-optional college admits nonsubmitters, and test optional becomes less appealing if its competitor goes test mandatory and makes the nonsubmission pool worse. In our substitutability example, a test-optional college does not admit the nonsubmitters, but the less-selective society wants it to admit them. So test optional becomes \emph{more} appealing for the college if its test-mandatory competitor worsens the nonsubmission pool.

Finally, a third example illustrates a distinct force that can push towards substitutability. Consider a college that prefers to be test optional when its competitor is test mandatory. If the competitor switches to test optional and rejects nonsubmitters, those rejected students may include some high-scoring ones. The college now has a higher yield for these desirable students. It reacts by going test mandatory and ``cherry-picking'' these students.

\section{Conclusion}
\label{sec:conclusion}

Our paper offers a resolution to the puzzle of why a college would choose to obtain less information about students by using a test-optional (or test-blind) admissions policy \citep{DFK24-PP}. 
We propose that going test optional helps a college alleviate social pressure regarding the students it admits. We introduce and solve a model of college admissions in which a college faces costs from making admission decisions that an external observer, society, disagrees with. Society is Bayesian and judges the college based on the available information, which is common to the college and society. 
The college commits to an imputation rule---stipulating, as a function of a student's observable characteristics, the test score assigned to non-submitters---and an acceptance rule specifying whether a student with any given observables and test score is admitted.

Our results in \autoref{sec:flexible} establish that when a college can flexibly choose its imputation rule, a test-optional regime is always weakly better for the college than a test-mandatory one. Test optional is often strictly better, reducing the college's cost from social pressure and/or delivering a student body it likes more. In \autoref{sec:restricted}, we study restricted imputation rules. Here, we find that going test optional may or may not benefit a college. For both flexible and restricted imputation, we identify which students benefit and which students are hurt by test optional. In \autoref{sec:AA} we explore an extended example of restricted imputation, illustrating that our framework can explain how a ban on affirmative action can result in a college choosing to go test blind. In \autoref{sec:competition}, we discuss how competition between colleges can affect the incentives to be test optional.

We close by discussing some alternative modeling assumptions and broader issues.

\subsection{Capacity Constraints and Interactions Across Students}

Our model assumes that a college admits any student that provides it a utility above some fixed threshold, normalized to zero. Modulo the caveat in \autoref{fn:capacity}, this abstracts away from capacity constraints: if a college accepts more applicants of one type, it may mechanically have to accept less applicants of other types.

Incorporating a capacity constraint would not, we believe, fundamentally change the point that test optional can be used to combat social pressure. But the analysis would be more complicated, because unlike in our model the college would have to maintain the same number of students under test optional as under test mandatory. Consequently, if students from one group benefit from test optional, then students from some other group would necessarily be harmed. This externality could raise important equity concerns in practice.

Our formulation also abstracts from the college's or society's preferences for one student depending on the admission outcome of another. In that sense, payoffs are additively separable across students. This omits considerations such as class balance or diversity (on legally-permissible dimensions).

\subsection{Alternative Restricted Imputation Rules}

The restricted imputation rule we have highlighted is that of no adverse inference: $\tau(x)=\E[t|x]$. There are at least two other rules that appear salient. 

First, colleges' claims to not punish non-submitters can be interpreted as a promise to impute missing test scores as equal to those of an average submitted score: $\tau(x)=\E[t|x,S=1]$. Notice that for any observable $x$, 
we cannot have a range of scores being submitted: a student with the lowest such score would, instead, not submit their score. Hence, this form of ``equal treatment'' effectively unravels to no student submitting their score.

Second, the reason that a college may not be able to flexibly impute missing scores is that it lacks commitment power, and instead it can only impute via Bayes rule: $\tau(x)=\E[t|x,S=0]$. Now, for any $x$, if there is a range of scores not being submitted, a student with the highest such score would instead submit. Hence, this imputation rule effectively unravels to every student submitting their score.

The upshot, then, is that under either of these alternative forms of restricted imputation, test optional would collapse to either test blind or mandatory.

\subsection{No Commitment to the Acceptance Rule}

Suppose the college cannot commit to its acceptance rule, instead admitting students ex-post optimally given their imputed/submitted scores.

In this case, a college's acceptance decision is simply determined by whether a student's imputed/submitted score is above or below the ex-post optimal bar, $\u t^*(x)$. Under flexible imputation, \autoref{prop:lessselective} implies that the outcome is unchanged at observables at which the college is less selective than society. But when the college is more selective, it can no longer set $\tau(x)>\u t^*(x)$ and reject non-submitters, which could have been optimal (\autoref{prop:moreselective}); the problem now is that the college must accept students who submit scores above $\u t^*(x)$. Consequently, the college now optimally sets $\tau(x)=\u t^*(x)$ and rejects non-submitters.

This means that a test-optional college now accepts all the students it would under test mandatory, and possibly additional ones (if, for some observable $x$ at which it is less selective than society, it chooses $\tau(x)=\infty$ and accepts all students with observable $x$). Hence, all students benefit, at least weakly, from test optional. Of course, this conclusion relies crucially on the college not having a capacity constraint.

\subsection{Non-Bayesian Society}
\label{sec:nonBayesian}

We have assumed that society is Bayesian. This means that if students from certain groups have lower test scores on average than others, society accounts for that in evaluating non-submitters. We view Bayesian updating as a way of tying our hands, showing that our mechanism goes through even when society can't be systematically misled. In practice, though, society might not take into account all information contained in the observables. For instance, if the observable vector $x = (x_0, x_1)$ has component $x_1$ corresponding to grades and component $x_0$ corresponding to race, society might evaluate non-submitters race-neutrally: $t^s = \mathbb{E}[t|x_1, S=0]$ rather than $t^s = \mathbb{E}[t|x_1,x_0, S=0]$, where $S=0$ indicates non-submission.

A college that faces such a non-Bayesian society may get an additional benefit from not seeing test scores. In particular, recall that in the specification of \autoref{sec:AA}, the college always prefers test mandatory when affirmative action is allowed. If society is instead constrained to race-neutral updating, one can show that there are parameters at which the college benefits from going test blind even when it can condition admissions on race. Intuitively,
when red types have higher test scores than green types ($p_r>p_g$), going test blind effectively makes society value green types more and red types less, on average. Going test blind thus brings the non-Bayesian society's preferences even closer to the college's.

\subsection{Non-Equilibrium Behavior by Students}

Our model assumes that students correctly understand a college's admissions policy---in particular, its imputation rule---and best respond in their score submission. In practice, students are often uncertain about whether to submit their score, leading to mistakes. With flexible imputation, such  mistakes can only reduce a college's benefit from being test optional; effectively, the college's problem transforms from one of flexible imputation to restricted imputation. The upshot can be that test optional does not actually benefit the college. 

Colleges' inability to fine-tune student submission behavior through their imputation rules provides another explanation, in addition to that after \autoref{prop:alignedprefs}, for why some colleges have recently reinstated their testing requirements. Indeed, one of the factors in Dartmouth's 2024 decision was that under test optional some applicants did not submit their scores even when submission would have ``helped that student tremendously, maybe tripling their chance of admissions,'' \href{https://www.thedartmouth.com/article/2024/02/college-to-reinstate-standardized-test-requirement-for-class-of-2029}{according to} Professor Bruce Sacerdote.

\subsection{Why Test Scores?}
\label{sec:test-scores} 

Our paper is silent as to why colleges choose to make test scores, rather than other application components, optional. One reason---outside of our model---may be that \emph{standardized} test scores are easier for society to evaluate, whereas the personal essay, the GPA at a particular high school, or extra-curricular achievements require more specialized expertise to evaluate. As such, while these other components may be informative of college success, they are subject to less outside scrutiny
and generate less disagreement costs.
 Indeed, every lawsuit opposing affirmative action has used standardized test scores as evidence of discrimination (\href{https://www.newyorker.com/news/daily-comment/the-supreme-court-appears-ready-finally-to-defeat-affirmative-action}{New Yorker, January 2022}).

\bibliographystyle{ecta}
\bibliography{DFK-bib}

\newpage

\appendix

\begin{centering}

\LARGE \textbf{Appendix}

\end{centering}

\section{Proofs}

\subsection{Proof for \autoref{sec:mandatory}}

\begin{proof}[Proof of \autoref{prop:mandatory}]
Under test mandatory, the student's score is always observed by both the college and society. So the college's problem for any observed $(x,t)$ is to choose $A\in \{0,1\}$ to maximize $$AU^c(x,t,t,1)+(1-A)U^c(x,t,t,0).$$
From the definition of the ex-post utility function $u^*(x,t)$ in \eqref{eq:u^*}, it is equivalent for the college to maximize $Au^*(x,t)$, which implies the result.
\end{proof}

\subsection{Proofs for \autoref{sec:analysis}}

\begin{proof}[Proof of \autoref{lem:prelim}]
Fix test optional with some imputation rule $\tau$. Consider a student with observable $x$ and imputed/submitted score $\hat t$. Given our assumption that the student submits if they have score $t>\tau(x)$ and does not submit if $t\leq \tau (x)$, whether the imputed/submitted score $\hat t$ is on path or off path depends only on the support of the score distribution $F_{t|x}$. If $\hat t$ is off path, then any college acceptance decision is optimal. Since any $\hat t<\tau(x)$ is necessarily off path, it is optimal to reject such $\hat t$. There are two remaining cases:

\begin{enumerate}
    \item $\hat t > \tau(x)$, and it is on path. Then the student must have submitted $\hat t$, and so by the logic of \autoref{prop:mandatory}, it is optimal for the college to accept the student if $u^*(x,\hat t)>0$ and reject the student otherwise.
    \item $\hat t = \tau(x)$, and it is on path. Then $\hat t$ is an imputed score. By similar reasoning to that in \autoref{prop:mandatory}, the expected utility gain from accepting these students types is proportional to the ex-post utility $u^*(x,L(\tau(x)|x))$, and so it is optimal to accept if that ex-post utility is positive and reject otherwise.
\end{enumerate}
We note that the resulting acceptance rule is monotonic, as $L(\tau(x)|x)\leq \tau(x)$ and $u^*(x,\cdot)$ is increasing.
\end{proof}

For clarity and notational ease, in the remainder of this Appendix section we assume that for each $x$, the cumulative distribution of test scores $F_{t|x}$ has a density $f(t|x)$. The proofs of Propositions \ref{prop:lessselective} and \ref{prop:moreselective} can be extended to arbitrary distributions of $t|x$ with additional notational burden.

\autoref{lem:breakupcosts} below will be used in the proofs of Propositions \ref{prop:lessselective} and \ref{prop:moreselective}.  In words, the lemma compares the disagreement from pooling  (e.g., via non-submission) versus separating (via submission) different groups of students, holding fixed the acceptance decisions. Part \ref{part:breakupincrease} says that breaking up one pool into two increases disagreement, at least weakly, which is then further increased by separating the higher pool. Part \ref{part:breakupconst} says that, if one pool is broken into two and society would make the same decision on both of the new pools---either rejecting both when $\mathbb{E}[t|x, t \in (\tau, \tau^h]] \le \u t^s(x)$, or accepting both when $\u t^s(x) \le L(\tau|x) \equiv \mathbb E [t| t \le \tau]$---then breaking up this pool does not in fact change disagreement costs. Part \ref{part:breakupseparate} establishes that turning a pooling region into a separating region doesn't change disagreement if society would make the same decision for all students in that range. Formally:

\begin{lemma}  \label{lem:breakupcosts}
Fix observables $x$. 
Given $-\infty \le \tau^l <\tau^h$, let $D^{\mathrm{pool}}(\tau^l, \tau^h; A)$ and $D^{\mathrm{sep}}(\tau^l, \tau^h; A)$ be the disagreement levels from making acceptance decision $A$ for all students with scores $t \in (\tau^l, \tau^h]$ while, respectively, pooling all students together or separating them:
\begin{align*}
D^{\mathrm{pool}}(\tau^l, \tau^h; A) & := \int_{\tau^l}^{\tau^h} d(x, \mathbb{E}\left[t | x, t \in (\tau^l, \tau^h]\right],A) f(t|x) \mathrm dt;
\\
D^{\mathrm{sep}}(\tau^l, \tau^h; A) & := \int_{\tau^l}^{\tau^h} d(x,t,A) f(t|x) \mathrm dt.
\end{align*}
\begin{enumerate}
\item \label{part:breakupincrease}  
Take any $\infty \le \tau \le \tau^h$. It holds that 
\begin{align*}
D^{\mathrm{pool}}( -\infty, \tau^h; A) & \le D^{\mathrm{pool}}( -\infty, \tau; A)  + D^{\mathrm{pool}}( \tau, \tau^h; A) \\ & \le  D^{\mathrm{pool}}( -\infty, \tau; )+ D^{\mathrm{sep}}( \tau, \tau^h).
\end{align*}
\item \label{part:breakupconst} 
 Take any $\tau <  \tau^h$. If $\mathbb{E}\left[t | x, t \in (\tau, \tau^h]\right] \le  \u t^s (x)$ or $\u t^s (x) \le L(\tau|x)$, then 
\[D^{\mathrm{pool}}(-\infty, \tau^h ; A) = D^{\mathrm{pool}}(-\infty, \tau ; A) + D^\mathrm{pool}(\tau, \tau^h; A).\]
\item \label{part:breakupseparate}
Take any $\tau^l < \tau^h$. If $\tau^h \le  \u t^s (x)$ or $\u t^s (x) \le \tau^l$, then 
\[
D^\mathrm{pool}(\tau^l, \tau^h;A) = D^\mathrm{sep}(\tau^l, \tau^h;A).
\]
\end{enumerate}

\end{lemma}

\begin{proof}[Proof of \autoref{lem:breakupcosts}]
Part \ref{part:breakupincrease} follows from convexity of the disagreement function $d(x, t^s, A)$ in $t^s$. Part \ref{part:breakupconst} and part \ref{part:breakupseparate} follow from the linearity of $d(x, t^s, A)$ on the domain $t^s \le \u t^s$ and on the domain $t^s \ge \u t^s$. For part \ref{part:breakupconst}, we also apply the fact that $L(\tau|x) \le 
\mathbb{E}[t | x, \tau < t \le \tau^h]$, and hence the assumptions guarantee that $L(\tau|x)$ and $\mathbb{E}[t | x, \tau < t \le \tau^h]$ are both on the same side of $\u t^s(x)$. For part \ref{part:breakupseparate}, the assumptions guarantee that $\tau^l$ and $\tau^h$ are both on the same side of $\u t^s$.
\end{proof}

\begin{proof}[Proof of \autoref{prop:lessselective}.] 
 Fix some observable $x$ at which the college is less selective than society. To reduce notation, the rest of this proof omits the $x$ argument in $\tau$, $\u t^c$, $\u t^*$, $\u t^s$, $f(t)$, and $L(t)$. Note that $L(t)$ is continuous in $t$ under the assumption we have made for this proof that the test score distribution has a density.

The college's payoff of imputing $\tau$ is constant over $\tau \in [-\infty, \u t^*]$: for any of these imputations, \autoref{lem:prelim} implies that the college rejects students with scores $t \le \u t^*$ and accepts those with scores $t>\u t^*$; since $\u t^*<\u t^s$, \autoref{lem:breakupcosts} (parts \ref{part:breakupconst} and \ref{part:breakupseparate}) implies that the disagreement cost does not change. 

To prove the result, then, it is sufficient to establish that the college's payoff from imputing $\tau \in [\u t^*, \infty]$ is decreasing and then increasing. In particular, take $t^\dagger\geq \u t^*$ such that $L(t^\dagger) = \u t^*$, with $t^\dagger = \infty$ if $L(t') < \u t^*$ for all $t'$. We will show that the college's expected payoff is decreasing in $\tau$ between $\u t^*$ and $t^\dagger$, then increasing in $\tau$ above $t^\dagger$.

To show that the college's payoff is decreasing in $\tau$ over the domain $\tau \in [\u t^*, t^\dagger)$, take some $\tau \in [\u t^*, t^\dagger)$. \autoref{lem:prelim} implies that it is optimal for the college to reject students with imputed/submitted score $\hat t \le  \tau$, since $L(\tau) \le L(t^\dagger)=\u t^*$. That is, non-submitters are rejected. Submitters with $t>\tau$ are admitted, since $\tau \ge \u t^*$. We now consider two cases: $\tau <  \u t^s$ and $\tau \ge \u t^s$.

\begin{itemize}
\item 
If $\tau <  \u t^s$, then the college's expected payoff $\mathbb{E}[A u^c(x,t) - \delta d(x, t^s, A)|x]$ can be written as 
\[\int_\tau^\infty u^c(x,t)f(t) \mathrm dt - \delta \int_\tau^{\u t^s} -u^s(x,t)f(t)\mathrm dt
\]
because the college only accepts students with $t >\tau$ and only incurs a disagreement cost for the students it accepts with $\tau < t \le \u t^s$. (There is no disagreement cost for rejecting the non-submission pool, which has $L(\tau) \le \u t^* < \u t^s$; and there is no disagreement cost for accepting students with $t \ge \u t^s$.) The right-derivative of this expected payoff with respect to $\tau$ is\footnote{We use the right-derivative because although the payoff is continuous in $\tau$ (under the assumption of continuous distributions) and differentiable almost everywhere, there are kinks at $\u t^*$ and $\u t^s$.}
\[
- (u^c(x,\tau) + \delta u^s(x,\tau)) f(\tau).
\]
This derivative is weakly negative because $ u^c (x,\tau) + \delta u^s(x,\tau) \ge 0$  for any $\tau \ge \u t^*$: $\tau \ge \u t^*$ implies  that $u^*(x, \tau ) \ge 0$, and $u^*(x, \tau )$ has the sign of  $ u^c(x,\tau) + \delta u^s(x,\tau)$.

\item If $\tau \ge  \u t^s$, then the college's expected payoff $\mathbb{E}[A u^c(x,t) - \delta d(x, t^s, A)|x]$ can be written as 
\[\int_\tau^\infty u^c(x,t)f(t) \mathrm dt
\]
because the college only accepts students with $t >\tau$ and does not incur a disagreement cost for any student. The right-derivative of this expected payoff with respect to $\tau$ is
\[-u^c(x, \tau) f(\tau),\]
which is weakly negative because $\tau \ge \u t^c$.
\end{itemize}

Next, we  show that the college's expected payoff is increasing in $\tau$ over the domain $\tau \in [t^\dagger, \infty]$. Note that when $\tau \in [t^\dagger, \infty]$, \autoref{lem:prelim} implies that it is optimal for the college to accept non-submitters (who have expected test score of $L(\tau) > \u t^*$) as well as submitters with $t>\tau$. Hence, all students are accepted. Moreover, pooling more students by raising $\tau$ always weakly reduces disagreement costs, if raising $\tau$ does not change acceptance decisions (\autoref{lem:breakupcosts} part \ref{part:breakupincrease}). Hence, raising $\tau$ over this domain weakly benefits the college.
\end{proof}

\begin{proof}[Proof of \autoref{prop:moreselective}.] 
 Fix some observable $x$ at which the college is more selective than society. For notational simplicity, the rest of this proof omits the $x$ argument in $\tau$, $\u t^c$, $\u t^*$, $\u t^s$, $f(t)$, and $L(t)$.

The college's payoff of imputing $\tau$ is increasing over $\tau \in [-\infty, \u t^*]$. To see this, observe that for any imputation in this range, \autoref{lem:prelim} implies that the college rejects students with scores $t \le \u t^*$ and accepts those with scores $t>\u t^*$;  \autoref{lem:breakupcosts} part \ref{part:breakupincrease} implies that the disagreement cost over the students with $t \le \u t^*$ decreases in the imputation level. (Increasing the imputation level corresponds to combining a pooling and a separating region into a single pooling region, while continuing to reject all students in that region.) Without loss, then, we can restrict attention to $\tau \ge \u t^*$.

Let $t^\circ$ be a test score at which $L(t^\circ)=\u t\Soc$, with $t^\circ = \infty$ if $L(t')<\u t\Soc$ for all $t'$ and $t^\circ =- \infty$ if $L(t')>\u t\Soc(x)$ for all $t'$. So for $t' \le t^\circ$, $L(t') \le \u t^s$; and for $t' \ge t^\circ$, $L(t') \ge \u t^s$. We will show that it is optimal for the college to set 
\[\tau=
\begin{cases}
\u t^* & \text{ if } t^\circ \leq \u t^*\\
t^\circ & \text{ if } t^\circ \in (\u t^*,\u t\Col)\\
\u t\Col & \text{ if } t^\circ \geq \u t\Col.
\end{cases}\]

\begin{itemize}
\item Suppose $t^\circ \leq \u t^*$. We seek to show that it is optimal to set $\tau = \u t^*$.

Under $\tau = \u t^*$, \autoref{lem:prelim} implies that the college rejects the pool of non-submitters with $t \le \u t^*$ (since $L(\u t^*) \le \u t^*$) and accepts submitters with $t > \u t^*$.  The college's expected payoff at $\tau = \u t^*$ can be written in the notation of \autoref{lem:breakupcosts} as 
\begin{align} \int_{\u t^*}^\infty u^c(x,t) f(t) \mathrm dt - \delta D^\mathrm{pool}(-\infty, \u t^*; A=0).
\label{eq:Payofftautstar}
\end{align}
 
Now consider, instead, $\tau = \tau^h > \u t^*$. There are two possibilities.

\underline{Case 1}: $L(\tau^h) \le \u t^*$. In this case,  the college rejects non-submitters (by \autoref{lem:prelim}), and so (applying \autoref{lem:breakupcosts} parts \ref{part:breakupconst} and \ref{part:breakupseparate}) the college's expected payoff can be written in the notation of \autoref{lem:breakupcosts} as 
\begin{align} \int_{\tau^h}^\infty &u^c(x,t) f(t) \mathrm dt  - \delta D^\mathrm{pool}(-\infty, \tau^h; A=0) \notag \\
&=\int_{\tau^h}^\infty u^c(x,t) f(t) \mathrm dt - \delta \left(D^\mathrm{pool}(-\infty, \u t^*; A=0) + D^\mathrm{sep}(\u t^*, \tau^h; A=0)\right)
 \label{eq:Payofftauh}\end{align}
The expected payoff of setting $\tau = \u t^*$ minus that of setting $\tau = \tau^h$ is given by the difference of expressions \eqref{eq:Payofftautstar} and \eqref{eq:Payofftauh}, which is
\[\int_{\u t^*}^{\tau^h} \left(u^c (x,t) + \delta u^s(x,t)\right) f(t) \mathrm dt .\]
We now observe that $u^c (x,t) + \delta u^s(x,t)$ has the sign of $u^*(x,t)$, which is positive on $t>\u t^*$, implying that the college prefers setting $\tau$ to $ \u t^*$ rather than $\tau^h$.

\underline{Case 2}:  $L(\tau^h) > \u t^*$. In this case, by \autoref{lem:prelim}, the college accepts non-submitters as well as the submitters with $t>\tau ^h$. That is, it accepts all students. Since $L(\tau^h)\geq L(\u t^*)\geq \u t^s$, it faces no disagreement costs, and its expected payoff is simply $\E[u^c(x,t)|x]$. Subtracting this from \eqref{eq:Payofftautstar} yields the expected payoff difference of setting $\tau = \u t^*$ and $\tau = \tau^h$:
\begin{align*}
-\int_{-\infty}^{\u t^*} &u^c(x,t) f(t) \mathrm dt - \delta \int_{-\infty}^{\u t^*} u^s(x, L(\u t^*))f(t)\mathrm dt\\
&=-\int_{-\infty}^{\u t^*} \left (u^c(x, L(\u t^*)) + \delta u^s(x, L(\u t^*))\right) f(t) \mathrm dt,
\end{align*}
where the equality is by the linearity of $u^c(x,t)$ in $t$.
We now observe that the above payoff difference is weakly positive because $u^c(x,t) + \delta u^s(x,t)$ is weakly negative for any $t \le \u t^*$, and because $L(\u t^*) \le t^*$. Hence, the college prefers setting $\tau$ to $ \u t^*$ rather than $\tau^h$.

\item Suppose $\u t^* < t^\circ < \u t^c$. We seek to show that it is optimal to set $\tau = t^\circ$.

At $\tau = t^\circ$, the college rejects non-submitters and faces no disagreement cost, so the college's expected payoff  $\mathbb{E}[A u^c(x,t) - \delta d(x, t^s, A)|x]$ at $\tau = t^\circ$ is

\begin{align}
\int_{t^\circ}^\infty u^c (x,t)  f(t) \mathrm dt.
\label{eq:Payofftautcirc}
\end{align}

At any $\tau=\tau^l\in [\u t^*, t^\circ)$, the college also rejects non-submitters and faces no disagreement cost, so its expected payoff is
\begin{align}\int_{\tau^l}^\infty u^c (x,t)  f(t) \mathrm dt,
\end{align}
which is clearly less than \eqref{eq:Payofftautcirc} because $u^c(x,t)<0$ on $(\tau^l, t^\circ)$. Hence the college prefers to set $\tau$ to $t^\circ$ over $\tau^l$.

Now consider setting $\tau = \tau^h > t^\circ$. There are two possibilities.

\underline{Case 1}: $L(\tau^h) \le \u t^*$. In this case, the college rejects the pool of non-submitters, and its expected payoff in the notation of \autoref{lem:breakupcosts} is
\begin{align}  
\allowdisplaybreaks
\int_{\tau^h}^\infty &u^c(x,t) f(t) \mathrm dt - \delta D^\mathrm{pool}(-\infty, \tau^h;A=0)  \notag
\\ &= \int_{\tau^h}^\infty u^c(x,t) f(t)\mathrm dt  - \delta \left (D^\mathrm{pool}(-\infty, t^\circ; A=0) + D^\mathrm{sep}(t^\circ, \tau^h; A=0) \right) \notag
\\ &= \int_{\tau^h}^\infty u^c(x,t) f(t)\mathrm dt  - \delta  D^\mathrm{sep}(t^\circ, \tau^h; A=0) \notag
\\ &= \int_{t^\circ}^\infty u^c(x,t) f(t)\mathrm dt -  \int_{t^\circ}^{\tau^h} \left( u^c(x,t)  + \delta u^s(x, t )\right) f(x|t) \mathrm dt \label{eq:Payofftauabovetcircreject},
\end{align}
where the first equality applies \autoref{lem:breakupcosts} parts \ref{part:breakupconst} and \ref{part:breakupseparate}; the second equality uses \mbox{$D^\mathrm{pool}(-\infty, t^\circ; A=0)=0$}, since $u^s(x,L(t^\circ)) = u^s(x, \u t^s) = 0$; and the third equality uses the definition of $D^\mathrm{sep}$. Observing that $u^c (x,t)  + \delta u^s(x, t )>0 $ on all $t > \u t^*$ implies that \eqref{eq:Payofftauabovetcircreject} is less than \eqref{eq:Payofftautcirc}.

\underline{Case 2}: $L(\tau^h) > \u t^*$. In this case, the college accepts the pool of non-submitters as well as the submitters and it pays no disagreement costs, so its expected payoff is 
\[\int_{-\infty}^{\infty} u^c (x,t) f(t)\mathrm dt = \int_{-\infty}^{t^\circ} u^c (x,t) f(t)\mathrm dt + \int_{t^\circ}^{\infty} u^c (x,t) f(t)\mathrm dt.
\]
This payoff is less than \eqref{eq:Payofftautcirc} since the first term is weakly negative because $t^\circ < \u t^c$.

\item Suppose $t^\circ \ge \u t^c$. We seek to show that it is optimal to set $\tau = \u t^c$.

The argument is straightforward: setting $\tau = \u t^c$ gives the college its first-best payoff. It admits students with $t > \u t^c$, and it rejects students with $t \le \u t^c$. The college faces zero disagreement cost for the accepted students, who all have  $t \ge \u t^c > \u t^s$. And the college also faces zero disagreement cost for the rejected pool of non-submitters, since $\u t^c \le t^\circ$ and therefore the pool has average test score $L(\u t^c) \le \u t^s$. \qedhere

\end{itemize}
\end{proof}

\begin{proof}[Proof of \autoref{prop:restrictedopt}]

Follows from \autoref{lem:prelim}.
\end{proof}

\begin{proof}[Proof of \autoref{prop:alignedprefs}]
Part \ref{aligned1} holds because the college's payoff from both test mandatory and test optional with the given imputation rule is continuous in $\delta$, and by assumption the college strictly prefers test mandatory when $\delta=0$. For part \ref{aligned2}, we fix society's preferences and only vary the college's. Following \autoref{fn:metric}, we write $u\Col\to u\Soc$ if $\sup_{x \in\mathcal{X}} |\u t\Col(x) - \u t\Soc(x)|\to 0$. It can be verified that, under either test mandatory or test optional with the given imputation rule, the college's payoff converges to its payoff when it shares society's preferences as $u\Col\to u\Soc$. The result follows from the hypothesis that the college would strictly prefer test mandatory if $u\Col=u\Soc$.
\end{proof}

\section{Imputation and Monotonic Acceptance Rules}
\label{sec:general_policies}

Let $\mathcal T:=\Reals \cup \{ns\}$, where $\Reals$ is the set of test scores and $ns$ denotes non-submission. A general policy for the college specifying when to admit a student, which we shall refer to as an \emph{allocation rule} to distinguish it from the admission rules considered in the main text, is a function $\pi:\mathcal X \times \mathcal T \to [0,1] $, where $\pi(x,t)$ is the probability of admitting a student with observables $x\in \mathcal X$ and score (non-)submission $t\in \mathcal T$. An \emph{outcome} under an allocation rule $\pi$ is any function $\mathcal X \times \Reals \mapsto [0,1]$ that obtains from composing $\pi$ with some student best response. Plainly, the outcome from any admission rule---a combination of an imputation rule $\tau$ and monotonic acceptance rule $\alpha$, as defined in \autoref{sec:primitives}---is the outcome under some allocation rule $\pi$.\footnote{Any admission rule has a unique outcome, given our assumption in the main text that the student submits their score if and only if it is strictly above the imputation. So, admissions rule $(\tau,\alpha)$ is outcome-equivalent to the allocation rule $\pi$ defined by $\pi(x,t)=\alpha(x,\tau(x))$ if $t\leq \tau(x)$ or $t=ns$, and $\pi(x,t)=\alpha(x,t)$ if $t>\tau(x)$.} Below, we explain why there is no loss of optimality for the college in restricting to admission rules: for any outcome under any allocation rule, there is an admissions rule whose outcome is weakly better for the college.

We start with two observations:
\begin{enumerate}
\item Without loss of optimality, the college can restrict to deterministic allocation rules, i.e., choose some 
$\pi:\mathcal X \times \mathcal T \to \{0,1\}$. This stems from the college's expected utility being linear in the admission probability, and the student's optimal action (submission or non-submission) only depending on the ordinal ranking of probabilities induced by the actions. See \citet{FK23deterministic} for further intuition and a proof.

\item (a) If $\pi(x,ns)=1$, then the college weakly prefers an outcome in which students with observable $x$ do not submit regardless of test score;\\ (b) if $\pi(x,ns)=\pi(x,t)=0$ for some score $t\in \Reals$, the college weakly prefers that students with observable $x$ and true score $t$ not submit.\\ The argument for both cases is that the outcome for the relevant students is unchanged if they don't submit---noting in part (a) that any outcome has students with observable $x$ admitted regardless of true test score---and the disagreement cost is weakly reduced by having students pool on non-submission.

\end{enumerate}

Now fix any observable $x$. The two observations imply that without loss of optimality, we can restrict to deterministic allocation rules 
and outcomes such that either (i) non-submitters are accepted and no student submits, or (ii) non-submitters are rejected and a student submits if and only if that leads to acceptance. Case (i) is outcome-equivalent to an admission rule with imputation $\tau(x)=\infty$ and acceptance $\alpha(x,\cdot)=1$, so we can focus on case (ii). 

Accordingly, suppose non-submitters are rejected, and a student submits if and only if that leads to acceptance. We will establish the claim that there is no loss of optimality in supposing the allocation rule is {monotonic}: given $x$ and any $t_L<t_H$, if $\pi(x,t_L)=1$ then $\pi(x,t_H)=1$.
Letting \mbox{$\bar{t}:=\sup\{t:\pi(x,t)=0\}$}, a monotonic allocation rule $\pi$ is outcome-equivalent to an admission rule with acceptance $\alpha(x,t)=1$ if and only if $t>\tau(x)$, where the imputation is $\tau(x)=\bar t$ if either the test score distribution $F_{t|x}$ is continuous, or $\bar t=-\infty$, or $\pi(x,\bar t)=0$;  
and $\tau(x)=\bar t-\epsilon$ for sufficiently small $\epsilon>0$ otherwise (i.e., if $F_{t|x}$ is discrete and $\pi(x,\bar t)=1$).\footnote{Recall from \autoref{fn:continuousdiscrete} that we assume a discrete $F_{t|x}$ has no accumulation points.}

To establish the monotonicity claim, first consider the case in which the distribution of test scores (at the fixed observable $x$) is continuous. Suppose there is a positive-probability set of scores that are accepted that are all lower than another positive-probability set of scores that are rejected. For any $t$, let $G_L(t)$ be the measure of accepted students with scores below $t$, and $G_H(t)$ be the measure of rejected students with scores above $t$. Since these are continuous functions and $\lim_{t\to -\infty} [G_L(t)-G_H(t)]<0<\lim_{t\to +\infty}[ G_L(t)-G_H(t)]$, there is $t'\in \Reals$ such that $G_L(t')=G_H(t')$. Consider a modification of the allocation rule to reject all scores below $t'$ and accept all scores above $t'$. The new outcome in which students submit if and only if their score is above $t'$ is weakly preferred by the college: it improves its underlying utility (because the accepted students are better) while reducing disagreement costs (because the rejected non-submitting students are worse). Therefore, the original allocation rule is improved by a monotonic one.

If the test score distribution is discrete, the same logic applies, except that we may require a public randomization device to allow splitting of the mass of students with a specific test score. Specifically, suppose that both the student's submission decision and the college's allocation rule can condition on a public random variable uniformly distributed on $[0,1]$. This allows for an arbitrary fraction of students at a certain test score to submit and be accepted, while the others with that score don't submit (and will be rejected if they do submit). The same logic as in the previous paragraph then applies.

\section{Details for Banning Affirmative Action}
\label{app:AA}

\subsection{A Model of Affirmative Action}

There are two potentially observable non-test dimensions, $x = (x_0, x_1)$. Dimension $x_0$ is binary, with realizations  in $\{r,g\}$ (red and green). Dimension $x_1$, which may represent some aggregate of GPA and/or extra-curricular achievement, takes continuous values in $\mathbb{R}$. Test scores are binary, with values normalized to $0$ and $1$.

The college and society have identical preferences over all factors except for the type dimension $x_0$. Society does not care about this dimension, but all else equal, the college wants to admit green types over red types.\footnote{We could allow for society to have preferences over a student's $x_0$ dimension as well; what is important is that the college favors green types more than society does.} Specifically, we assume that
\begin{align*}
u\Soc(x,t) &= x_1 +  t, \\
u\Col(x,t) &= x_1 +  t + \beta  \ind_{x_0 = g}-c,
\end{align*}
with  $\beta>c>0$, and $\ind_{x_0=g}$ an indicator for green types. The parameter $\beta$ is the bonus the college gives to green types over red types. The parameter $c$ is not essential to our analysis, but it allows for the college and society to have different test-score bars for both red and green students. It can be interpreted as the (opportunity) cost for a college of admitting any student. We have normalized the analogous constant in society's utility to zero.  The assumption $\beta>c>0$  implies that the college has a lower test-score bar than society for green types and a higher one for red types. Note that the the college's ex-post utility is
\[
u^*(x,t) =  x_1 +  t + \frac{\beta}{1+\delta} \ind_{x_0 = g} - \frac{c}{1+\delta}.
\]

Let $x_0=g$ with probability $q \in (0,1)$ and $x_0=r$ with probability $1-q$. We assume that the distribution of test scores depends on $x$ only through $x_0$: $\Pr(t=1 | x=(x_0,x_1))=p_{x_0} \in (0,1)$. Our primary interest is in the case of $p_r > p_g$, meaning that green types, which are favored by the college, have a worse distribution of test scores. This may correspond to green students being an underrepresented demographic group, for instance. But we also allow for the opposite case of $p_r<p_g$, in which the college's favored group has a better test score distribution. 
Here, green students may correspond to those from rich families, who have better access to test preparation, and are favored by the college because of donor considerations. If the green students correspond to legacy applicants, it may be that either $p_r<p_g$ or $p_r>p_g$.

We take $x_1$ to be independent of both $x_0$ and $t$. We also assume that $x_1$ is uniformly distributed over a large enough interval. Specifically, $x_1 \sim U[\underline x_1, \overline x_1]$, with $\underline x_1 < c - \beta - 1$ and $\overline x_1 > c$. The inequality on $\underline x_1$ guarantees that there are students with $x_1$ low enough  that neither the college nor society wants to admit them, even if they are otherwise as desirable as possible ($x_0=g$ and $t=1$). The inequality on $\overline x_1$ guarantees that there are students with $x_1$ high enough that the college and society want to admit them even if they are otherwise as undesirable as possible ($x_0=r$ and $t=0$).

We will consider the college's choice over whether to be test mandatory or test blind in two observability regimes. First, we allow both dimensions of $x$ to be observable, which we call \emph{affirmative action allowed}. Then we consider only $x_1$ to be observable, with the dimension $x_0$ unobservable; we call this regime \emph{affirmative action banned}. We interpret the switch from the first to the second regime as a policy change where society---which does not intrinsically care about $x_0$---bans the use of that dimension in admissions. 
This may represent a law or court decision forbidding the use of race or legacy status in admissions.\footnote{Note that we assume that when $x_0$ is unobservable to the college, it is also unobservable to society. While society does not value $x_0$ directly, the observability of $x_0$ to society could still matter for the calculation of the college's social costs. This is because, if society can observe $x_0$ but cannot observe test scores, then it would expect a different test score for green students ($p_g$) than red students ($p_r$). We assume that a law preventing the college from making inferences of this form also stop society from making/penalizing the college based on such inferences.}

\subsection{Results}

\paragraph{Affirmative action allowed.} \label{sec:AAresults} Consider first the case when affirmative action is allowed. 

Under test mandatory, the college can choose a distinct threshold of $x_1$ above which to admit students at each $(x_0,t)$ pair.\footnote{Since we will be comparing test mandatory with test blind, it turns out to be convenient for our analysis to take the perspective of $x_1$ admissions thresholds rather than test score thresholds.} This threshold is determined by setting the ex-post utility to $0$. Since the college favors green students, its $x_1$ threshold  will be lower by $\beta/(1+\delta)$ for green students than for red students at each score level $t$. From society's perspective, the college uses an $x_1$ threshold that is too low for green students  and too high for red students---but crucially, the gap between society's preferred threshold and what the college uses does not vary with $t$.\footnote{The gap is $(\beta-c)/(1+\delta)$ for green students and $c/(1+\delta)$ for red students.}

Under test blind, the college chooses an admissions threshold on dimension $x_1$ that depends on the student's type $x_0$ but not the test score $t$. However, $x_0$ is informative about $t$: the college and society evaluate students of type $x_0$ as if they have the expected test score $\mathbb{E}[t|x_0] = p_{x_0}$. If $p_r>p_g$, the college's preference for green students is countered by the fact that green students have lower test scores on average than red students. So the college will now use a lower $x_1$ threshold for green students than red students only if its preference parameter $\beta$ is sufficiently large: specifically, if and only if $\beta/(1+\delta)>p_{r}-p_{b}$. Regardless, the gap between the college's chosen $x_1$ threshold and society's preferred threshold is the same as under test mandatory, for any test score $t$---that gap did not depend on the test score, and utilities are linear in the test score. 

We can establish: 

\begin{proposition}
If affirmative action is allowed, then the college  prefers test mandatory to test blind. \label{prop:aaallowed}
\end{proposition}

The reason is that going test blind leads to a set of students that the college prefers less, but in the current specification there is never a countervailing benefit of reducing disagreement cost.
The latter point stems from two sources. First, as noted above, for any given $x_0$ type (and test score, under test mandatory), the gap between society's preferred $x_1$ threshold and what the college uses is independent of the regime, even though these thresholds do shift across regimes. Second, our assumption of a uniform distribution of $x_1$ means that the total disagreement cost for students of a given $x_0$ type (at a given test score, or averaging over test scores) only depends on the size of the gap.

\paragraph{Affirmative action banned.} Now consider the case when affirmative action is banned. 

Under test mandatory, the observed test score is informative about a student's type $x_0$. Specifically, since there are a fraction $q$ of green types in the population and the probability of test score $t=0$ for a student of type $x_0$ is $1-p_{x_0}$, we compute the probability of a student being green conditional on $t=0$ as
\[P_g^0 := \Pr(x_0=g|t=0) =  \frac{q}{q+(1-q) \frac{1-p_r}{1-p_g}}.\]
Analogously, conditional on $t=1$, the probability of a green type is \[P_g^1 := \Pr(x_0=g|t=1) =  \frac{q}{q+(1-q) \frac{p_r}{p_g}}.\]
Let $\Delta := P_g^0-P_g^1$ be the difference between these two quantities, i.e., a low test score implies a $\Delta$ higher probability of $x_0=g$ than a high test score. Note that $\Delta>0$ if $p_r>p_g$, whereas $\Delta<0$ if $p_r<p_g$. Based on the inference of $x_0$ from $t$, the college's underlying utility gives a bonus of $\beta \Delta$ to students with low test scores relative to those with high scores. 
As a result, the college now values a high test score  $1-\beta \Delta$ units higher than a low score, whereas society still values it 1 unit higher. That is, unlike when affirmative action is allowed, the gap between society's preferred $x_1$ admissions threshold and what the college chooses now varies with the test score.\footnote{Absent affirmative action, it is as if the college's underlying utility from a student is $x_1+t+\beta P^t_g-c$, and so the college's gain from a student with test score $t=1$ over $t=0$ is $1+\beta P^1_g-\beta P^0_g = 1-\beta\Delta$. Given its underlying utility, the college's ex-post utility from a student is $x_1+t+\frac{\beta P^t_g-c}{1+\delta}$. The gap between the college's chosen $x_1$ admissions threshold with society's preference is the term $\frac{\beta P^t_g-c}{1+\delta}$, which varies with $t$ so long as $P^0_g \neq P^1_g$, or equivalently $\Delta \neq 0$.}  We impose the assumption that $\beta \Delta < 1$, so the college still prefers students with higher test scores.

There is now an avenue for test blind to help the college. Under test blind, since the college evaluates all students as having $\Pr(x_0=g) = q$ and $\mathbb{E}[t]=q p_g + (1-q)p_r$, it is as if the college's utility from any student is $x_1+\E[t]+q\beta-c$. Analogously, it is as if society's utility from any student is $x_1+\E[t]$. If $c=q\beta$, which means the college and the society seek to admit the same number of students overall, then it is as if their utilities agree, and the college implements its preferred admissions policy---subject to being test blind and no affirmative action---at zero disagreement cost. More generally, the disagreement cost is always lower under test blind than test mandatory. Whether the reduced disagreement cost outweighs the allocative loss from being test blind depends on parameters, specifically the intensity of social pressure $\delta$ and the college's bonus to low-scoring students $\beta \Delta$.

\begin{proposition}
Suppose affirmative action is banned. 
 If $(1+\delta)(2 \beta \Delta -1) \ge (\beta \Delta)^2$, then the college prefers test blind, and otherwise the college  prefers test mandatory.
\label{prop:aabanned}
\end{proposition}
Recall we assume $\beta \Delta < 1$.  \autoref{prop:aabanned} implies that if $\beta \Delta \le 1/2$, the college always prefers test mandatory: the allocative losses (``admission mistakes'') from not observing test scores are larger than those from simply implementing society's preferred decision rule and incurring no disagreement. When  $\beta \Delta \in (1/2,1)$, there is a trade-off, and test blind will be preferred if the intensity of social pressure, $\delta$, is sufficiently large.  The following corollary develops this and other comparative statics.

\begin{corollary} \label{cor:aacompstats}
Suppose that affirmative action is banned ($x_0$ is unobservable) and that a low test score is associated with $x_0=g$  ($\Delta >0$).
\begin{enumerate}

\item 
There is some $\beta^*\in \left(\frac{1}{2\Delta},\frac{1}{\Delta}\right)$ such that the college prefers test mandatory when $\beta < \beta^*$ and prefers test blind when $\beta>\beta^*$.
\item 
There is some $\Delta^*\in \left(\frac{1}{2\beta},\frac{1}{\beta}\right)$ such that the college prefers test mandatory when $\Delta < \Delta^*$ and prefers test blind when $\Delta>\Delta^*$.
\item If $\beta \Delta \le 1/2$, then the college prefers test mandatory for all $\delta$; if $\beta \Delta \in (1/2,1)$, then there is some $\delta^*>0$ such that the college prefers test mandatory when $\delta<\delta^*$ and prefers test blind when $\delta>\delta^*$.
\end{enumerate}
\end{corollary}

\subsection{Society's Preferences}

We now consider society's payoff under different affirmative action and testing regimes. Society's realized utility for an individual student is $A u^s(x,t)$, where the dummy variable $A$ indicates whether the student is admitted. We assume that society's objective is to maximize its expected utility across the pool of applicants.

\begin{proposition}\label{prop:society} Society's preferences over affirmative action and testing regimes are as follows:
\begin{enumerate}

\item Fixing the testing regime as mandatory or blind, society prefers banning affirmative action to allowing affirmative action. \label{part:socpayoffaa}

\item Fixing affirmative action as banned or allowed, society prefers test mandatory to test blind. \label{part:socpayofftest}

\item \label{part:stackelberg} Suppose society chooses the affirmative action regime and then the college chooses the testing regime. Then banning affirmative action can harm society. In particular, if $\beta \Delta \in \left(1/2,1\right)$, there exist thresholds $0<\underline\delta\leq \bar \delta<\infty$ such that 
 (i) if affirmative action is banned, the college chooses test blind if $\delta >\underline \delta$, and (ii) society is harmed by banning affirmative action if $\delta > \bar \delta$, while it benefits if $\delta < \bar \delta$.\footnote{If $\beta \Delta \le  1/2$, the college never goes test blind, and so, by part \ref{part:socpayoffaa} of the proposition, society always benefits from banning affirmative action.}

\end{enumerate}
\end{proposition}

The first two parts of the proposition are intuitive, since society does not want the admission decision to depend on whether a student is red or green (which suggests part \ref{part:socpayoffaa}) but does want the decision to depend on the test score (which suggests part \ref{part:socpayofftest}). If society could choose both the testing and affirmative action regimes, it would ban affirmative action and choose test mandatory. However, part \ref{part:stackelberg} of the proposition cautions that if society chooses the affirmative action regime and the college subsequently chooses the testing regime, society can be worse off by banning affirmative action. Specifically, when $\delta$ is large enough, banning affirmative action backfires because the college's response of going test optional results in a student pool that society likes less than under test mandatory and affirmative action allowed. Indeed, as $\delta$ gets arbitrarily large, society's payoff is arbitrarily close to society's first best when affirmative action is allowed and there is mandatory testing, while it is bounded away when affirmative is banned and the college goes test blind. But when $\delta$ is intermediate (between the thresholds $\underline \delta$ and $\overline \delta$ in \autoref{prop:society} part \ref{part:stackelberg}), society is better off by banning affirmative even though it results in the college going test blind.\footnote{It is possible that $\u \delta = \o \delta$, in which case whenever a ban on affirmative action leads to test optional, society is harmed by the affirmative-action ban.}

\subsection{Proofs for Results on Banning Affirmative Action} 

As a preliminary observation, we can write the college's loss relative to first best as its allocative loss plus the cost of social pressure. At a given $(x_0, t)$ pair of test scores and group memberships, the assumption of a uniform distribution over $x_1$ implies that the college's allocative loss depends only on the difference between the college's chosen $x_1$-cutoff for admission and the college's ideal $x_1$ cutoff. Specifically, let $f := \frac{1}{\o x_1 - \u x_1}$ be the (constant) density of the $x_1$ distribution on its support. If the college's chosen cutoff is $r$ above its ideal cutoff, then its allocative loss on this $(x_0, t)$ pair is 
\begin{align} \label{eq:squareloss}
\int_0^r f x \mathrm{d}x = \frac{f}{2} r^2.
\end{align}
Society's (allocative) loss is given by the same formula, when the chosen cutoff is $r$ above society's preferred cutoff.

\begin{proof}[Proof of \autoref{prop:aaallowed}]
Suppose that affirmative action is allowed. Here, there is no interaction between the college's decisions at different realizations of $x_0$. So, it suffices to show that test mandatory would be preferred to test blind for any fixed $x_0=x'_0$ in $\{r, b\}$. 

Fixing $x_0=x'_0$, let $h := u^c(x'_0, x_1, t) - u^s(x'_0, x_1, t) = \beta \ind_{x'_0 = g} -c$ be the difference between the college's and society's utility for admitting a student of type $x_0=x'_0$, which does not depend on $x_1$ or $t$. It then holds that $u^c(x'_0, x_1, t) - u^*(x'_0, x_1, t)  = \frac{\delta}{1+\delta}h$, and that $u^*(x'_0, x_1, t) - u^s(x'_0, x_1, t) = \frac{1}{1+\delta}h$. 
Given its information, the college sets $x_1$ admissions cutoffs at the value of $x_1$ setting the expectation of $u^*(x'_0, x_1, t)$ to $0$. Note that the college's ideal $x_1$-cutoff for students in group $x_0=x'_0$ with test score $t$ is  $-t-h$, whereas society's ideal $x_1$-cutoff is $-t$.

\paragraph{The college's loss under test mandatory.}

At $(x'_0, t)$, the college's chosen $x_1$-cutoff for admission is $ \frac{\delta}{1+\delta}h$ above its ideal point, yielding allocative loss (from \eqref{eq:squareloss}) of 
\begin{align} \label{eq:allowmandalloc}
\frac{f}{2} \frac{h^2 \delta^2}{(1+\delta)^2}.
\end{align}
Similarly, the college's chosen $x_1$-cutoff for admission is $ - \frac{1}{1+\delta} h$ above society's ideal point, leading to an allocative loss for society of $\frac{f}{2} \frac{h^2}{(1+\delta)^2}$. The college then pays a social pressure cost equal to $\delta$ times that, or 
\begin{align}\frac{f}{2} \frac{\delta h^2}{(1+\delta)^2}. \label{eq:allowmandsoc} 
\end{align}
Both of these expressions are independent of $t$, meaning that these expressions also represent the college's losses averaged over test scores.

The college's loss under test mandatory, for students with $x_0=x'_0$, is the sum of \eqref{eq:allowmandalloc} and \eqref{eq:allowmandsoc}.

\paragraph{The college's loss under test blind.}

With unobservable test scores, the players evaluate students of type $x_0=x'_0$ as if they have the expected test score of  $p_{x'_0}$. The college's chosen $x_1$ cutoff for students of type $x_0=x'_0$ sets $u^*(x'_0, x_1, p_{x'_0})$ to $0$, i.e., a cutoff of $x_1=-p_{x'_0} - \frac{h}{1+\delta}$.

To calculate the college's allocative losses, we compare the college's chosen (test-independent) $x_1$ admissions cutoffs to its (test-dependent) ideal cutoffs. Recall that the college's ideal cutoff at test score $t$ is $x_1 = -t - h$. So at $t=1$, the college's chosen cutoff is $1-p_{x'_0} + \frac{\delta}{1+\delta} h$ above its ideal point; at $t=0$, the college's chosen cutoff is $-p_{x'_0} + \frac{\delta}{1+\delta}h$ above its ideal point.
The college's expected allocative loss over test scores, once again plugging into \eqref{eq:squareloss}, is therefore given by
\begin{align}
p_{x'_0} \frac{f}{2} \left (1-p_{x'_0} + \frac{\delta}{1+\delta}h\right)^2  +(1-p_{x'_0}) \frac{f}{2} \left(-p_{x'_0} + \frac{\delta}{1+\delta}h\right)^2 \notag \\ = \frac{f}{2} \frac{h^2 \delta^2}{(1+\delta)^2} + \frac{f}{2} p_{x'_0}(1-p_{x'_0}). \label{eq:allowoptalloc}
\end{align}

To calculate social costs, we compare the college's chosen $x_1$ admissions cutoff not to society's ideal cutoff, but to society's preferred cutoff given that test scores are not observed. Society's preferred $x_1$-cutoff is given by $-p_{x'_0}$. The chosen cutoff is $- \frac{h}{1+\delta}$ above society's preferred cutoff. We can now plug into \eqref{eq:squareloss} to calculate society's loss relative to its preferred cutoff (given its information) as $\frac{f}{2} \frac{h^2}{(1+\delta)^2}$. The college's social pressure cost is $\delta$ times that, or 
\begin{align}\frac{f}{2} \frac{\delta h^2}{(1+\delta)^2}. \label{eq:allowoptsoc} \end{align}

The college's  loss under test blind, for students with $x_0=x'_0$, is the sum of \eqref{eq:allowoptalloc} and \eqref{eq:allowoptsoc}. 

\paragraph{Comparison.}

Comparing expressions \eqref{eq:allowmandsoc} and \eqref{eq:allowoptsoc}, the social pressure cost under test blind is identical to that under test mandatory. Comparing expressions \eqref{eq:allowmandalloc} and \eqref{eq:allowoptalloc}, the allocative loss is higher under test blind. Hence, the college prefers test mandatory.
\end{proof}

\begin{proof}[Proof of \autoref{prop:aabanned}]
Suppose that affirmative action is banned.  Let $ET := \mathbb{E}[t] = q p_r + (1-q) p_g$ be the average test score in the population, i.e., the share with test score $t=1$. Recall that $P_g^t = Pr(x_0=g|t)$. We will now calculate the college's loss in each testing regime.

In each case, we will evaluate the college's allocative loss relative to a benchmark where the college must make decisions independently of the unobservable $x_0$ type. The college's ideal cutoff at test score $t$, given that it must pool together students across the two $x_0$ types, is  $-t - \beta P_g^t +c$. 

\paragraph{The college's loss under test mandatory.} Society's ideal $x_1$-cutoff for admitting a student of with test score $t$ is $-t$. The college's chosen cutoff, setting the expected ex post utility to 0, is $-t - \frac{1}{1+\delta} (\beta P_g^t - c)$.

To calculate the allocative loss, observe that the college's chosen cutoff is $\frac{\delta}{1+\delta}(\beta P_g^t - c)$ above its ideal cutoff at test score $t$. Plugging into \eqref{eq:squareloss}, its allocative loss across the two test scores is given by
\begin{align} \label{eq:testsnoaaalloc}
& (1-ET) \frac{f}{2} \left(\frac{\delta}{1+\delta}(\beta P_g^0 - c)\right)^2 + ET \frac{f}{2} \left(\frac{\delta}{1+\delta}(\beta P_g^1 - c)\right)^2.
\end{align}

To calculate the loss due to social pressure, observe that the chosen $x_1$-cutoff is $-\frac{1}{1+\delta}(\beta P_g^t - c)$ above society's preferred cutoff. The college's expected loss due to social pressure (plugging this difference into \eqref{eq:squareloss} for each test score, taking expectation over test scores to find society's loss, and then multiplying by $\delta$) is therefore
\begin{align}\label{eq:testsnoaasoc} 
& \delta (1-ET)\frac{f}{2} \left(\frac{1}{1+\delta} (\beta P_g^0 -c)\right)^2 + \delta ET \frac{f}{2} \left(\frac{1}{1+\delta} (\beta P_g^1 -c)\right)^2 
.
\end{align}

The college's total loss is \eqref{eq:testsnoaaalloc} plus \eqref{eq:testsnoaasoc}. 

\paragraph{The college's loss under test blind.} The average test score is $ET$, and so society's preferred $x_1$-cutoff is $-ET$. The college's chosen cutoff, setting the expected ex post utility to 0, is $-ET  - \frac{1}{1+\delta}(\beta q -c) $, where $q$ is the probability of $x_0=g$.

Again, we calculate the college's allocative loss relative to its ideal point with observable $t$ but unobservable $x_0$. At test score $t$, the chosen cutoff minus the ideal cutoff is 
\begin{align*} 
t - ET + \beta P_g^t - \frac{q \beta}{1+\delta} - \frac{c \delta }{1+\delta}
\end{align*}
Plugging into \eqref{eq:squareloss} and taking the expectation across test scores, the college's allocative loss is  given by
\begin{align} \label{eq:notestsnoaaalloc}
 (1-ET) &\frac{f}{2}\left ( - ET + \beta P_g^0 - \frac{q \beta}{1+\delta} - \frac{c \delta }{1+\delta}\right)^2
+ ET \frac{f}{2} \left( 1 - ET + \beta P_g^1 - \frac{q \beta}{1+\delta} - \frac{c \delta }{1+\delta} \right)^2.
\end{align}

The difference between the college's chosen cutoff and society's preferred cutoff is \mbox{$- \frac{1}{1+\delta} (\beta q - c)$}. Plugging into \eqref{eq:squareloss} and multiplying by $\delta$, the college's loss from social pressure is
\begin{align}
\frac{f}{2} \frac{ \delta (\beta q - c)^2}{(1+\delta)^2}. \label{eq:notestsnoaasoc}
\end{align}

The college's total loss is \eqref{eq:notestsnoaaalloc} plus \eqref{eq:notestsnoaasoc}.

\paragraph{Comparison.} 
The net benefit of choosing test blind rather than test mandatory is given by the loss from test mandatory minus the loss from test blind, i.e.,
\[\eqref{eq:testsnoaaalloc} + \eqref{eq:testsnoaasoc} - \eqref{eq:notestsnoaaalloc} - \eqref{eq:notestsnoaasoc}.\] Substituting in $q = (ET)P_g^1 + (1-ET) P_g^0$ and $\Delta = P_g^0 - P_g^1$ and then simplifying, we can rewrite this net benefit as
\[
\frac{f}{2}\frac{ET (1-ET)}{1+ \delta} \left( (1+\delta) (2 \beta \Delta - 1) - (\beta \Delta)^2\right).
\]
The above expression is weakly positive if and only if $(1+\delta) (2 \beta \Delta - 1) \ge (\beta \Delta)^2$.
\end{proof}

\begin{proof}[Proof of \autoref{cor:aacompstats}.]
Suppose that affirmative action is banned. \autoref{prop:aabanned} establishes that the college prefers test blind to test mandatory if and only if 
\begin{align}(1+\delta)(2 \beta \Delta -1) \ge (\beta \Delta)^2. \label{eq:whentb} \end{align}
Recall we maintain the assumptions that $\beta>0$, $\beta \Delta <1$, and for this corollary, $\Delta>0$. We prove each part of the corollary in turn:

\begin{enumerate}

\item Rewriting \eqref{eq:whentb}, the college prefers test blind if and only if 
\[-\Delta^2 \beta^2 +2 \Delta (1+\delta) \beta - (1+\delta)\ge 0.\]
The LHS is a concave quadratic that is negative at $\beta = \frac{1}{2 \Delta}$ (equal to $-1/4$) and positive at $\beta = \frac{1}{\Delta}$ (equal to $ \delta$). Hence, there exists $\beta^* \in (\frac{1}{2 \Delta}, \frac{1}{\Delta})$ such that the college prefers test blind when $\beta > \beta^*$ and test mandatory when $\beta < \beta^*$. Using the quadratic formula, $\beta^* = \frac{1+ \delta - \sqrt{\delta (1+\delta)}}{\Delta}$.

\item Since \eqref{eq:whentb} is symmetric with respect to $\beta$ and $\Delta$, the argument of the previous part goes through unchanged after swapping $\beta$ and $\Delta$. We get $\Delta^* = \frac{1+\delta - \sqrt{\delta (1+\delta)}}{\beta}$.

\item If $\beta \Delta \in (0,1/2)$, then the LHS of \eqref{eq:whentb} is nonpositive and the RHS is strictly positive, implying that test mandatory is optimal.

If $\beta \Delta > 1/2$, then we can rewrite \eqref{eq:whentb} as $\delta \ge \frac{(1-\beta \Delta)^2}{2 \beta \Delta -1}$, and hence the result holds for $\delta^*  = \frac{(1-\beta \Delta)^2}{2 \beta \Delta -1}>0$. \qedhere
\end{enumerate}
\end{proof}

\begin{proof}[Proof of \autoref{prop:society}.]

As in \eqref{eq:squareloss}, at a given $(x_0, t)$, society's loss relative to its first best when the college's chosen $x_1$-cutoff for admission is $r$ above society's ideal cutoff is $\int_0^r f x \mathrm{d}x = \frac{f}{2} r^2$.  Society's expected loss across all values of $x_0$ and $t$ is equal to the expectation of $ \frac{f}{2} r^2$ over the distribution of $r$, with $r$ the difference between the chosen cutoff (which may depend on $x_0$ and $t$) and society's ideal cutoff (which depends only on $t$). 
Since the loss $\frac{f}{2} r^2$ is convex in $r$, mean-preserving spreads in the distribution of these cutoff differences make society worse off.

\underline{Part \ref{part:socpayoffaa}}. Fix any testing regime. The distribution of cutoffs at each test score when affirmative action is allowed is a mean-preserving spread of the distribution when affirmative action is banned. Hence, society prefers banning affirmative action.

\underline{Part \ref{part:socpayofftest}}. First, suppose that affirmative action is allowed. Fix some type $x_0=x'_0$, at which the college has a utility bonus of $h := u^c (x'_0,x_1, t) - u^s(x'_0, x_1, t) = \beta \ind_{x'_0=g}-c$ relative to society. Under test mandatory,  at each test score, the chosen $x_1$-cutoff is $\frac{h}{1+\delta}$ above society's ideal cutoff. Under test blind, at $t=1$, the chosen cutoff is $1-p_{x'_0} +\frac{h}{1+\delta}$ above society's cutoff; and at $t=0$, the chosen cutoff is $-p_{x'_0} +\frac{h}{1+\delta}$ above society's cutoff. Hence, under test blind, at each type $x'_0$, the distribution of the chosen cutoff minus society's cutoff is given by
\[\begin{cases}
1-p_{x'_0} +\frac{h}{1+\delta} &\text{with probability } p_{x'_0} \\
-p_{x'_0} +\frac{h}{1+\delta} &\text{with probability } 1- p_{x'_0}.
\end{cases}\]
This distribution is a mean-preserving spread of the constant $\frac{h}{1+\delta}$. Hence, society is worse off under test blind for each realization $x'_0$ of $x_0$, and so is worse off in expectation.

Next, suppose that affirmative action is banned. As also defined in the proof of \autoref{prop:aabanned}, we let $ET := \mathbb{E}[t] = q p_r + (1-q) p_g$ denote the average test score in the population, i.e., the share of students with test score $t=1$. At test score $t$, the college's ideal $x_1$-cutoff is $-t - \beta P_g^t +c$ (recall $P_g^t = \Pr(x_0=g | t)$), and society's ideal $x_1$-cutoff is $-t$.

Under test mandatory with affirmative action banned, the college's chosen $x_1$-cutoff is $- \frac{1}{1+\delta} ( \beta P_g^t-c)$ above society's ideal point at test score $t$. That is, a share $ET$ of students have cutoffs $- \frac{1}{1+\delta} ( \beta P_g^{1}-c)$ above society's ideal point, and a share $1-ET$ have cutoffs $ -\frac{1}{1+\delta} (\beta P_g^{0}-c)$ above. Plugging in $q = (ET)P_g^1 + (1-ET) P_g^0$ and $\Delta = P_g^0 - P_g^1$, some algebra yields that the distribution of chosen cutoffs minus society's ideal cutoffs is
\begin{align} \label{eq:MandAAbanDist}
\begin{cases}
-\frac{1}{1+\delta} (\beta q -c) + (1-ET)\frac{\beta \Delta}{1+\delta} &\text{with probability } ET
\\
-\frac{1}{1+\delta} (\beta q - c) - ET \frac{ \beta \Delta }{1+\delta} &\text{with probability } 1- ET.
\end{cases}
\end{align}

Under test blind with affirmative action banned, the college's chosen $x_1$ cutoff is $-ET -\frac{1}{1+\delta}(\beta q-c) $. This means that for the $ET$ share of students with $t=1$, the chosen $x_1$-cutoff is $-\frac{1}{1+\delta}(\beta q-c) + (1 - ET)$ above society's ideal cutoff of $-1$; for the $1-ET$ share with $t=0$, the chosen cutoff is $-\frac{1}{1+\delta}(\beta q-c ) - ET $ above society's ideal cutoff of $0$. That is, the distribution of chosen cutoffs minus society's ideal cutoffs is
\begin{align} \label{eq:BlindAAbanDist}
\begin{cases}
-\frac{1}{1+\delta} (\beta q -c) + (1-ET) &\text{with probability } ET
\\
\frac{1}{1+\delta} (\beta q- c) - ET  &\text{with probability } 1- ET.
\end{cases}
\end{align}
Since $\beta \Delta < 1$ (by assumption) and $1+\delta>1$, the distribution in \eqref{eq:BlindAAbanDist} is a mean-preserving spread of that in \eqref{eq:MandAAbanDist}. Hence, when affirmative action is banned, society prefers test mandatory to test blind. 

\underline{Part \ref{part:stackelberg}}. From \autoref{prop:aabanned}, if $(1+\delta)(2\beta \Delta-1)<(\beta  \Delta)^2$, then the college chooses test mandatory under an affirmative action ban. If $(1+\delta)(2\beta \Delta-1)>(\beta  \Delta)^2$, which implies $\beta \Delta >1/2$, the college chooses test blind under an affirmative action ban. 

So, when $\beta \Delta \in (0, 1/2]$, society prefers to ban affirmative action: it prefers test mandatory and no affirmative action to test mandatory with affirmative action (by part \ref{part:socpayoffaa}).

Now  suppose that $\beta \Delta > 1/2$. Let $\underline \delta:= \frac{(\beta \Delta)^2}{2 \beta \Delta - 1}-1$ be the solution to  $(1+\delta)(2\beta \Delta-1)=(\beta  \Delta)^2$. For $\delta < \u \delta$, the college chooses test mandatory, in which case society prefers to ban affirmative action. For $\delta>\u \delta$, the college chooses test blind.
In this case, we need to compare society's payoff of test mandatory with affirmative action versus test blind without affirmative action.

The distribution of chosen $x_1$-cutoffs minus society ideal cutoffs under test mandatory with affirmative action is 
\[
\begin{cases}
\frac{\beta-c}{1+\delta}  &\text{with probability } q
\\
\frac{-c}{1+\delta} &\text{with probability } 1-q.
\end{cases}
\]
Society's corresponding payoff loss is
\begin{align} \label{eq:soclossNoNo} \frac{f}{2 (1+\delta)^2} \left(c^2 - 2 c q \beta + q \beta^2\right).\end{align}

The distribution of cutoffs minus society ideal points under test blind without affirmative action is given by \eqref{eq:BlindAAbanDist}. Society's payoff loss is correspondingly
\begin{align}\frac{f}{2 (1+\delta)^2} \left((\beta q - c)^2 + (1-ET) ET (1+\delta)^2\right)\label{eq:soclossYesYes} \end{align}
with $ET = q p_g + (1-q) p_r$.

The sign of\eqref{eq:soclossYesYes} minus \eqref{eq:soclossNoNo} tells us whether society prefers test mandatory with affirmative action or test blind without affirmative action. The sign of that difference is the same as the sign of 
\mbox{$ET(1-ET)(1+\delta)^2 - q (1-q) \beta^2$}. 
This expression equals zero when $\delta$ equals $$\delta ' := \beta \sqrt{\frac{q (1-q)}{ET(1-ET)}}-1.$$ 
When $\delta > \delta'$, society prefers test mandatory with affirmative action to test blind without affirmative action; when $\delta < \delta'$, the preference is reversed.

Finally, let $\o \delta := \max \{\u \delta, \delta'\}$. We now see that (i) when $\delta > \u \delta$, the college chooses test blind if affirmative action is banned; and (ii) taking into account the college's response in choosing its testing regime, society prefers to ban affirmative action if $\delta < \o \delta$, and prefers to allow affirmative action if $\delta > \o \delta$.
\end{proof}

\section{Competition Examples}
\label{app:competition}

This section provides three numerical examples to substantiate the discussion in \autoref{sec:competition} about competition. 

All three examples have two colleges. There is a single observable, and thus we omit the dependence of all variables on $x$. At this observable, students have test scores uniformly distributed between 0 and 100. 
Two identical colleges have underlying utility $u^c(t) = t-\underline t^c$; society has utility $u^s(t) =t-\underline t^s$; and the colleges place a weight $\delta = 1$ on social pressure, implying ex-post utilities $u^*(t)=t-\underline t^*$ with $\underline t^* = (\underline t^c+\underline t^s)/2$. A test-optional college is restricted to impute $\tau = 50$, the average test score, for nonsubmitters. If admitted by both colleges, a student chooses betwen them uniformly at random.

\begin{example}[Complements due to adverse selection]
Let $\underline t^c=5$ and $\underline t^s=25$, implying $\underline t^*=15$. Here, society is more selective than the college. Our calculations below  will establish strategic complementarity. That is, when a college's competitor is test mandatory, the college prefers to be test mandatory; but when a college's competitor is test optional, the college prefers to be test optional.

A test-mandatory college admits students with $t>\underline t^* = 15$ and rejects students with $t \le 15$. Assume that a test-optional college admits nonsubmitters with $t \in [0,50]$, which is optimal so long as the yield-weighted average test score of nonsubmitters is above $\underline t^*=15$, as will be the case. The test-optional college also admits the submitting students with $t>50$.

A test-mandatory college facing another test-mandatory college has yield of $1/2$ for all of the students it admits, since the other college makes identical admission decisions. The college then gets underlying utility of $\frac{1}{2}\int_{15}^{100} \frac{1}{100}(t-5) \mathrm dt = 22.3125$, social pressure costs of $\frac{1}{2}\int_{15}^{25}\frac{1}{100}(25-t)\mathrm t = .25$, and a net payoff of $ \simeq 22.1$. 

A test-mandatory college facing a test-optional college also has yield of $1/2$ for all of the students it admits, because the other college admits all students. So its payoff is also $\simeq22.1$.

A test-optional college facing another test-optional college has yield of $1/2$ for all students. The yield-weighted average test score for nonsubmitters is just the unweighted expectation of $25$. The college's underlying utility from admitting every student is $\frac{1}{2}\int_{0}^{100}\frac{1}{100}(t-5)\mathrm dt  = 22.5$, and social pressure costs are 0. So its payoff is $22.5$.

Finally, a test-optional college facing a test-mandatory college has a yield of $1/2$ for students with $t>15$, and a yield of $1$ for students with $t\le 15$. The yield-weighted average test score for nonsubmitters is $20.9615$.\footnote{The average test score between 0 and 15 is 7.5; the average test score between 15 and 50 is 32.5; and the weighted average, putting a weight of $1/2$ on test scores between 15 and 50, is $(7.5 \cdot (15-0) + 1/2 \cdot 32.5 \cdot (50-15)) / ((15-0) + 1/2 \cdot (50-15))$.} The college's underlying utility from admitting every student is $22.5$, and social pressure costs are $(25-20.9615) \cdot \frac{1}{100} \cdot  (1 \cdot (15-0) + \frac{1}{2}(50-15))=1.3125$. Its payoff is $\simeq 21.2$. 

We see that when a college's competitor is test mandatory, the college prefers to be test mandatory ($22.1>21.2$). When a college's competitor is test optional, the college prefers to be test optional ($22.5>22.1$). We also see that a test-optional college prefers its competitor to be test-optional ($22.5>21.2$).
 
\end{example}

\begin{example}[Substitutes due to adverse selection] \label{ex:substitutes_adverse}
Let $\underline t^c=50$ and $\underline t^s=20$, implying $\underline t^*=35$. Here, the college is more selective than society. Our calculations below  will establish strategic substitutability. That is, when a college's competitor is test mandatory, the college prefers to be test optional; but when a college's competitor is test optional, the college prefers to be test mandatory.

A college that is test mandatory admits students with $t>\underline t^* = 35$ and rejects students with $t \le 35$. Assume that a test-optional college rejects nonsubmitters with $t \in [0,50]$, which is optimal so long as the yield-weighted average test score of nonsubmitters is below $\underline t^*=35$, as will be the case. The test-optional college also admits the submitting students with $t>50$.

A test-mandatory college facing another test-mandatory college has yield of $1/2$ for students with $t>35$ and yield of $1$ for students with $t\le 35$, since the other college makes identical admission decisions. The college then gets underlying utility of $\frac{1}{2}\int_{35}^{100} \frac{1}{100}(t-50) \mathrm dt = 5.6875$, social pressure costs of $\int_{20}^{35}\frac{1}{100}(t-20)\mathrm dt = 1.125$, and a net payoff of $\simeq 4.6$. 

A test-mandatory college facing a test-optional college has a yield of $1/2$ for students with $t>50$ and a yield of 1 for students with $t\le 50$. So it gets underlying utility of $\int_{35}^{50} \frac{1}{100}(t-50) \mathrm dt+\frac{1}{2}\int_{50}^{100} \frac{1}{100}(t-50) \mathrm dt = 5.125$, social pressure costs of $\int_{20}^{35}\frac{1}{100}(t-20) \mathrm dt = 1.125$, and a net payoff of $4$. 

A test-optional college facing another test-optional college has a yield of $1/2$ for students with $t>50$ and yield of 1 for students with $t\le 50$. The yield-weighted average test score for nonsubmitters is just the unweighted expectation of $25$. The college's underlying utility from rejecting nonsubmitters and accepting submitters with $t>50$ is $\frac{1}{2}\int_{50}^{100}\frac{1}{100}(t-50)\mathrm dt  = 6.25$, and social pressure costs are $\frac{1}{100} (25-20) (50-0) = 2.5$. So its payoff is $3.75$.

Finally, a test-optional college facing a test-mandatory college has a yield of $1/2$ for students with $t>35$, and a yield of $1$ for students with $t\le 35$. The yield-weighted average test score for nonsubmitters is $21.9118$.\footnote{The average test score between 0 and 35 is 17.5; the average test score between 35 and 50 is 42.5; and the weighted average, putting a weight of $1/2$ on test scores between 35 and 50, is $(17.5 \cdot (35-0) + 1/2 \cdot 42.5 \cdot (50-35)) / ((35-0) + 1/2 \cdot (50-35))$.} The college's underlying utility from rejecting nonsubmitters and accepting submitters is $6.25$, as above, and social pressure costs are $(21.9118-20) \cdot \frac{1}{100} \cdot (1\cdot (35-0) + 1/2 \cdot (50-35))=.812515$. Its payoff is $\simeq 5.4$. 

We see that when a college's competitor is test mandatory, the college prefers to be test optional at this observable ($5.4>4.6$). When a college's competitor is test optional, the college prefers to be test mandatory ($4>3.75$). We also see that a test-optional college prefers its competitor to be test-mandatory ($5.4>3.75$).
\end{example}

\begin{example}[Substitutes due to cherry picking]
Let $\underline t^c=40$ and $\underline t^s=25$, implying $\underline t^*=32.5$. Here, the college is again more selective than society. As in \autoref{ex:substitutes_adverse}, our calculations below will establish strategic substitutability, but the mechanism now owes to ``cherry-picking'' rather than adverse selection.

A college that is test mandatory admits students with $t>\underline t^* = 32.5$ and rejects students with $t \le 32.5$. As in \autoref{ex:substitutes_adverse}, a test-optional college rejects nonsubmitters with $t \in [0,50]$ and admits the submitting students with $t>50$.

In contrast to \autoref{ex:substitutes_adverse}, the payoff of a test-optional college is now independent of its competitor's testing regime: regardless of whether the competitior is test mandatory or test optional, the college gets a yield of $1/2$ for the submitters that it admits, and it faces no social pressure costs for the nonsubmitters that it rejects. (Society's bar is 25; and the yield-weighted average test score is $25$ when the competitor is test optional, and it is below $25$ when the competitor is test mandatory.) The test-optional college's payoff is thus $\frac{1}{2}\int_{50}^{100}\frac{1}{100}(t-40)\mathrm dt  = 8.75$.

A test-mandatory college facing a test-mandatory competitor has a yield of $1/2$ for students with $t>32.5$ and a yield of $1$ for students with $t\le 32.5$, since the other college makes identical admission decisions. The college then gets underlying utility of $\frac{1}{2}\int_{32.5}^{100} \frac{1}{100}(t-40) \mathrm dt = 8.86$, social pressure costs of $\int_{25}^{32.5}\frac{1}{100}(t-25)\mathrm dt = 0.28$, and a net payoff of $\simeq 8.58$. 

A test-mandatory college facing a test-optional competitor has a yield of $1/2$ for students with $t>50$ and yield of 1 for students with $t\le 50$. So it gets underlying utility of $\int_{32.5}^{50} \frac{1}{100}(t-40) \mathrm dt+\frac{1}{2}\int_{50}^{100} \frac{1}{100}(t-40) \mathrm dt = 8.97$, social pressure costs of $\int_{25}^{32.5}\frac{1}{100}(t-25)\mathrm dt = 0.28$, and a net payoff of $8.79$. 

We see that, as in \autoref{ex:substitutes_adverse},  when a college's competitor is test mandatory, this college prefers to be test optional ($8.75>8.58$). When a college's competitor is test optional, this college prefers to be test mandatory ($8.79>8.75$). We also see that a test-mandatory college prefers its competitor to be test-optional ($8.79>8.58$), because that allows it to cherry-pick the students with $t \in(32.5,50)$---whom it wants to admit, on average---without competition.  
\end{example}

\end{document}